\documentclass[oribibl,envcountsame]{llncs}

\usepackage{graphicx}
\RequirePackage{bussproofs/bussproofs}

\usepackage{color}
\usepackage[lined]{algorithm2e}

\usepackage{subfigure}
\usepackage{verbatim}
\usepackage{array}

\usepackage{amsmath}
\usepackage{amssymb}
\usepackage{cite}
\usepackage{url}
\usepackage{setspace}
\usepackage{fancyhdr}
\usepackage[toc,page,titletoc]{appendix}
\usepackage{tabularx}
\usepackage{wrapfig}

\usepackage{turnstile}

\usepackage{tikz}

\usetikzlibrary{shapes,snakes}

\usetikzlibrary{fadings}
\tikzstyle{uedge}=[draw=blue!50!red]
\tikzstyle{fedge}=[draw=blue]
\tikzstyle{iedge}=[draw=red]
\tikzstyle{redge}=[draw=green!50!black]
\tikzstyle{rnode}=[draw,inner sep=2pt,color=black]
\tikzstyle{tnode}=[circle,minimum width=3pt,fill,inner sep=0pt]
\tikzstyle{dotnode}=[circle,minimum width=2pt,fill,inner sep=0pt]
\tikzstyle{labn}=[font=\sffamily,circle,fill=white,inner sep=1pt,draw=black]
\tikzstyle{legn}=[font=\scriptsize]
\tikzstyle{reln}=[circle,fill=white,inner sep=.4pt,draw=black]
\tikzstyle{oreln}=[circle,fill=white,inner sep=.4pt,draw=black!50,solid]
\tikzstyle{oree}=[thick,draw=black!50,densely dashed]
\tikzstyle{ree}=[thick,draw=black]
\tikzstyle{calcn}=[rectangle%
                   ,rounded corners=2mm%
                   ,thick%
                   ,draw=black%
                   ,top color=white,bottom color=black!20,%
                   minimum height=.5cm,minimum width=.5cm,inner sep=2pt%
                 ]
\tikzstyle{expcalcn}=[rectangle%
                   ,thick%
                   ,draw=green!40!black
                   ,top color=white,bottom color=green!20,%
                   minimum height=.5cm,minimum width=.5cm,inner sep=2pt%
                 ]

\usepackage{subfigure}
\usepackage{verbatim}
\usepackage{array}
\usepackage{amsmath}
\usepackage{amssymb}
\usepackage{amsthm}
\usepackage{url}
\usepackage[total={5.9in,9.6in}, top=1.0in, inner=1.35in]{geometry}
\usepackage{setspace}
\usepackage{fancyhdr}

\usepackage{turnstile}
\usepackage[toc,page,titletoc]{appendix}

\allowdisplaybreaks

\newcommand{\Q}{\ensuremath{\mathcal{Q}}\xspace}
\newcommand{\QBF}[2]{\ensuremath{{#1}\, \text{.} \,{#2}}\xspace}

\newtheorem{observation}{Observation}

\usepackage{xspace}

\newcommand{\res}{\textsf{Res}\xspace}
\newcommand{\rest}{\textsf{Res}$_\textsf{T}$\xspace}
\newcommand{\unit}{\textsf{unit-Res}\xspace}
\newcommand{\sld}{\textsf{SLD-Res}\xspace}
\newcommand{\sldres}{\textsf{SLD-Res}\xspace}

\newcommand{\slin}{\textsf{s-linear Res}\xspace}

\newcommand{\tlin}{\textsf{t-linear Res}\xspace}
\newcommand{\slres}{\textsf{SL-res}\xspace}

\newcommand{\QRAT}{\textsf{QRAT}\xspace}
\newcommand{\QIOR}{\textsf{QIOR}\xspace}

\newcommand{\QIORplus}{\textsf{QIOR}$^+$\xspace}

\newcommand{\QRATplus}{\textsf{QRAT}$^+$\xspace}

\newcommand{\qrc}{\textsf{Q-Res}\xspace}

\newcommand{\qurc}{\textsf{QU-Res}\xspace}

\newcommand{\sldQ}{\textsf{SLD-Q-Res}\xspace}

\newcommand{\derivesunit}{{\scriptsize \sststile{\textsf{unit-Res}\xspace}{}}}
\newcommand{\derivessld}{{\scriptsize \sststile{\textsf{SLD-Res}\xspace}{}}}
\newcommand{\derivessldQ}{{\scriptsize \sststile{\textsf{SLD-Q-Res}\xspace}{}}}

\newcommand{\derivessldQmulti}{ \,\scalebox{0.8}[0.8]{$\sststile{\textsf{SLD-Q}_\textsf{multi}\textsf{-Res}\xspace}{}$}\,}

\title{Extending Prolog for Quantified Boolean Horn Formulas 
}

\author{
 Anish Mallick\inst{1}
\and {Anil Shukla}\inst{2}
}

\institute{%
  Pontifica Universidad Catolica de Chile, Chile.\\
  \email{anish.mallick@mat.uc.cl}
  \and
  Department of Computer Science and Engineering, IIT Ropar, India.\\
  \email{anilshukla@iitrpr.ac.in}
}

\begin{document}
\maketitle

\begin{abstract}
Prolog is a well known declarative programming language based on propositional Horn formulas. It is useful in various areas, including artificial intelligence, automated theorem proving, mathematical logic and so on. An active research area for many years is to extend Prolog to larger classes of logic. Some important extensions of it includes the constraint logic programming, and the object oriented logic programming. However, it cannot solve problems having arbitrary quantified Horn formulas. 

To be precise, the facts, rules and queries in Prolog are not allowed to have arbitrary quantified variables. The paper overcomes this major limitations of Prolog by extending it for the quantified Boolean Horn formulas. We achieved this by extending the SLD-resolution (\sld) proof system for quantified Boolean Horn formulas, followed by proposing an efficient model for implementation. The paper shows that the proposed implementation also supports the first-order predicate Horn logic with arbitrary quantified variables. 

The paper also introduces for the first time, a declarative programming for the quantified Boolean Horn formulas. 

\end{abstract}

\section{Introduction}
\label{sec:intro}
Prolog is a declarative programming language developed by Alain Colmerauer, Bob Pasero and Philippe Roussel in 1972~\cite{Roussel75,Roussel93}. It was first designed for solving problems related to natural language processing. Very soon, arithmetic aspects of logic was added to it, using the ideas from Robert Kowalski SL-resolution prover~\cite{Kowalski71}. Today, Prolog is useful not only in the areas of language processing and automated theorem proving, but also in the areas of artificial intelligence, relational database, and design automation to name a few.

Prolog is, in a way an interactive system, which supports a kind of man-machine conversation. That is, a user is allowed to specify some relevant knowledge and subsequently ask it valid questions regarding the same. 
For example, one may specify several facts about a family, say family members and their relationships, along with some specific rules, for example a rule about siblings. 
Then one may ask if two elements in the space are siblings?

In Prolog, facts and rules are defined first then query follows. 
Though all facts and rules have to be predicate horn formulas with quantifiers, any rules defined have to be univerally quantified. 
For example, 
\begin{align*}
\forall x,y~ female(x)\wedge parent(x,y)\rightarrow mother(x,y)
\end{align*}
is a valid rule but 
\begin{align*}
\forall x,y\exists z~ parent(z,x)\wedge parent(z,y)\rightarrow sibling(z,y)
\end{align*}
is not. \\
On the other hand for queries we are allowed logical Horn clauses with existential quantifiers only. For example,
\begin{align*}
\hspace*{6mm} \exists x. parent(x,bob)
\end{align*}
is a valid query but
\begin{align*}
\hspace*{6mm} \exists x \forall y. ancestor(x,y)
\end{align*}
is not.
The primary focus of this work is to remove these limitations (see, \cite[Chapter 14]{Rowe1988} for reference) and extend the scope of Prolog to handle larger class of problems.
We also propose an efficient algorithm to the problem
\begin{align*}
\hspace*{6mm} \text{Is } \mathcal{F}=\Pi.\phi \implies \Pi.(\phi \wedge C)?,
\end{align*}
where, $\mathcal{F}$ is a first-order predicate Horn formula, and $C$ a first-order predicate Horn clause. Note that the above problem is in general undecidable~\cite{Halpern91, Voigt19}. The efficiency lies in the fact that the algorithm will detect loop or infinite braching in linear time and stop.
It should be clear that we are not claiming that any SLD-Q-resolution (Definition~\ref{def:sldq-res}) will not halt if this algorithm outputs loop, 
just that it will determine possiblity of not halting following the current course of action.
So this should not be viewed as an attempt towards halting problem but a contigency. 
These are explained further in the following section.

\subsection{Our Contributions}
\label{subsec:contri}
In this section we summarize the contribution of this article.
\begin{enumerate}
\item {\bf Extending Prolog for quantified Boolean Horn formulas}. A Prolog algorithm can be expressed as logical part plus control part. Logical part is expressed as facts and rules. The facts and rules are represented as (propositional) Horn formulas (Section \ref{subsec:propHorn}). The control part is based on the \sld proof system (Definition \ref{def:sld-resolution}). In particular, once a query clause $C$ is presented by the user, Prolog control tries to find an \sld refutation of $\neg C$, and accordingly return the answer. 

The article extended the \sld proof systems for the quantified Boolean Horn formulas and defined the \sldQ proof system (Section \ref{sec:sldQ}). Quantifed Boolean Horn formulas are all QBFs in closed prenix form with CNFs in which every clause has at most one positive literal (Definition~\ref{def:QHorn}). The article shows that \sldQ is sound and complete for the quantified Boolean Horn formulas and also proposes an efficient implementation for the same. That is, given a quantified Boolean Horn formula $\mathcal{F}=\Pi.\phi$ as facts and rules, and a quantified Horn query clause $C$, the paper proposes an efficient algorithm (Algorithm~\ref{linear-ProloAlgo-for-QBFs}) to answer the following: 
\begin{align*}
\text{Is } \mathcal{F}=\Pi.\phi \implies \Pi.(\phi \wedge C)?
\end{align*}
Obviuosly, one may encode verification problems as a quantified Boolean Horn formulas~\cite{verification-QHorn-2013, verification-QHorn2-2013,MonniauxG16}, which in general is undecidable. Our proposed algorithm is efficient in the sense that it can detect a possible loop during evaluation efficiently.

As a result, the article allows us to specify facts and rules as a quantified Boolean Horn formulas with quantified query clause. Refer Section~\ref{subsec:app1} for an application.

The article shows how our efficient algorithm (Algorithm~\ref{linear-ProloAlgo-for-QBFs}) handles the first-order predicate Horn logic as well. To be precise, the article shows that given facts, rules and queries as first-order predicate Horn clauses, Algorithm~\ref{linear-ProloAlgo-for-QBFs} solves efficiently the above implication problem.

\item {\bf Overcomes the weaknesses of existing Prolog}. As already discussed, Prolog does not allow arbitrary quantified variables in facts, rules and queries. All variables appearing in facts and rules are universally quantified and all variables appearing in the query are existentially quantified. These limitations arise due the Prolog control (engine) part, which is based on \sld proof systems for Horn formulas. 

This article overcomes the limitations of the existing Prolog by allowing arbitrary quantified variables in facts, rules and queries. As a result several problems which were impossible to be solved via existing Prolog can now be attempted. We present few examples below after few comments. 

The encoding of the problem statement of the Example~\ref{example:nocycle-tree} and Example~\ref{example-bipartite} (below) uses only universally quantified variables and can be specified in the existing Prolog. However, the query of the corresponding problem cannot be specified in the existing Prolog as it requires universally quantified variables. 

The encoding of the problem statement and the query of the Example~\ref{example:simple-relations} (below) uses both the existentially and universally quantified variables, and so cannot be specified in the existing Prolog. However, the problem is easily handled by the extended Prolog (QBF-Prolog, Section~\ref{sec:qbf-prolog}) of this article.

\begin{example}\label{example:nocycle-tree}
{\bf Tree-bipartite-problem}: Consider the problem related to trees and bipartite graphs. Recall, a graph $G=(V,E)$ is said to be bipartite iff its vertex set $V$ can be partitioned into two parts $U$ and $W$ such that every edge in $E$ connects a vetex in $U$ to a vertex in $W$. Consider the following: undirected graphs without cycle are bipartite, and trees are acyclic. These facts can be encoded as rules in the existing Prolog. However, consider the following query: are trees bipartite? 

The query required universal variables, and hence cannot be handled by the existing Prolog. However, the extended Prolog of the paper is capable to handle the problem. For the detailed encoding and solution of the problem, refer Section~\ref{subsec:tree-bipartite}.
\end{example}

\begin{example}\label{example-bipartite}
{\bf Bipartite problem}: Consider a Prolog program which specifies the set of all bipartite graphs $G$.  One can easily encode bipartite graphs in first-order predicate Horn logic using only universally quantified variables. Hence, this can be encoded in the existing Prolog. But, consider the following query regarding a bipartite graph $G=(\{U,W\},E)$: 
pick a vertex $x$ belonging to $U$ and move to some vertex $y$ using an edge $\{x,y\}$ and from $y$ again move to a vertex $z$ using the edge $\{y,z\}$. Is the vertex $z \in U$? 

Obviously, the above query is valid and needed universal variables in the encoding. Hence, the existing Prolog does not supports such queries. However, the present paper extend Prolog to handle such queries as well. For the detailed encoding of the bipartite problem, refer Section~\ref{subsec:app2}.
\end{example}

\begin{example}\label{example:simple-relations}
{\bf Simple relations}: The problem uses the following predicates: 
\begin{align*}
\notag & P(h,k)~~\text{/* the predicate is 1 iff } k = 2h */\\
\notag & R(h,k)~~\text{/* the predicate is 1 iff } h<k */
\end{align*}
We may give several interpretations to these predicates. For example, if $k$ and $h$ represent graphs, then the predicate $R(k,h)$ is $1$ iff the graph $h$ is a subgraph of the graph $k$. The predicate $P(k,h)$ is $1$ iff the graph $k$ is a superset of the graph $h$. That is, the graph $k$ is constructed from $h$ by say adding a vertex.

Let us now consider the following rules:
\begin{align*}
\notag & \forall h, k. P(h,k) \rightarrow R(h,k)~~\text{/* rule 1 */}\\
\notag & \forall h_1,h_2,h_3. [R(h_1,h_2) \wedge R(h_2,h_3)] \rightarrow R(h_1,h_3)~~\text{ /* rule 2 */}
\end{align*}
Now consider the following query:
\begin{align*}
\notag & \forall h \exists g,k. [ P(h,g) \wedge P(g,k)] \rightarrow R(h,k)~~\text{ /* query */}
\end{align*}
Clearly the query is correct according to our interpretation. We show that this problem can be easily solved by the QBF-Prolog of this article. 
For the detailed encoding and solution of the problem, refer Section~\ref{subsec:simple-relations}.
\end{example}

Finally, we show an example (Example~\ref{example-tree}) which cannot be handled by the existing Prolog, but also cannot be solved via the extended Prolog. The proofs of such problems are induction based, however, our algorithm is unable to mimic the inductive proof and just return a loop as the output. This is not surprising, even the pigeonhole principle has a short inductive proof, but is hard for resolution~\cite{haken_1985}. Consider the example below.

\begin{example}\label{example-tree}
Consider a Prolog program which specifies the set of all trees, with a unique root and such that every child has a unique parent. Consider the following query: is the tree connected? Existing prolog unables to handle such programs. QBF-Prolog, proposed in this article just return a $\textit{loop}$ for this problem. For the detailed encoding of the tree and query, refer Section~\ref{sec:tree-encoding}.
\end{example}

\item {\bf An interactive QBF-solver for the quantified Boolean Horn formulas (QBF-Prolog, Section~\ref{sec:qbf-prolog})}. There exists several QBF-solver in the literature (Section~\ref{subsec:qbf-solver}), then why QBF-Prolog? 

First of all, QBF-Prolog, deals with the quantified Horn formulas. 
It is well known that several problems of practical importance, such as program verifications~\cite{verification-QHorn-2013, verification-QHorn2-2013,MonniauxG16}, can be encoded as quantified Horn formulas. It is important to design an efficient QBF solver for the same. 
QBF-Prolog is such a solver. It uses the structure of the quantified Boolean Horn formulas and solves the satisfiability problem in linear time. (Of course, the QBF-Prolog may return a loop as well, but it detects the same in linear time).

QBF-Prolog can also be used for the quantified renamable Horn formulas.
We say that a QBF $\mathcal{F}$ is a quantified renamable Horn formula, if $\mathcal{F}$ can be transformed into a quantified Horn formula by negating every instance of one or more of its variables.
For example, the QBF $\exists x_1 \forall x_2 \exists x_3. (x_1 \vee x_2 \vee x_3) \wedge (\neg x_1 \vee x_2 \vee \neg x_3)$ is a quantified renamable Horn formula, since by complementing $x_1$ and $x_2$ we get the following quantified Boolean Horn formula: $\exists x_1 \forall x_2 \exists x_3. (\neg x_1 \vee \neg x_2 \vee x_3) \wedge ( x_1 \vee \neg x_2 \vee \neg x_3)$. This operation of choosing a subset of variables and replacing each positive literal of such a variable by the corresponding negative literal and vice versa, is called {\bf renaming}. It is easy to observe that renaming preserves satisfiability. That is, a QBF $\mathcal{F}$ is satisfiable if and only if the QBF $\mathcal{F}'$ obtained via renaming is satisfiable.

Given a QBF formula $\mathcal{F}$ it is possible to determine in linear time whether $\mathcal{F}$ is quantified renamable Horn formula~\cite{Hebrard94}. The algorithm also gives the set of variables which needs to be complemented in order to get a quantified Horn formula. Thus, using the algorithm from~\cite{Hebrard94}, and via the QBF-Prolog, one can solve the satisfiable problem for a quantified renamable Horn formula efficiently. Note that, \cite{Hebrard94} solves the problem for the propositional case, however, the same algorithm works for the quantified formulas as well, we just need to ignore the quantifiers.

Another advantage of QBF-Prolog is it's interactive nature: As it is based on Prolog, it supports man-machine conversation. To the best of our knowledge, all existing QBF solvers are imperative in nature. Also, QBF satisfiability problem is well suited to be solve via backtracking and hence is very natural to pick Prolog and extend it for the QBFs. 

\end{enumerate}

\subsection{Organisation of the paper}
\label{subsec:organization}
The remainder of the paper is organized as follows. We review the basic notations and preliminaries in Section~\ref{sec:notations}. 
In Section~\ref{sec:sldQ}, we extend the \sld proof systems for the quantified Boolean Horn formulas and proof the completeness and soundness of the new proof systems for the same. We extend Prolog for the quantified Boolean formulas, with the restrictions that the query clause contains no new variables, in Section~\ref{sec:qbf-prolog}, and also propose an efficient implementation model for the same. We extend Prolog for quantified Boolean Horn formulas without any restriction and for first-order predicate Horn formulas in Section~\ref{sec:qbf-prolog-general}. In Section~\ref{sec:application-of-qbf-prolog}, we present some applications of the QBF-Prolog. Finally, we present conclusions in Section~\ref{sec:conclusion}. 
 
\section{Notations and Preliminaries}
\label{sec:notations}
Quantified Boolean formulas (QBFs) are extensions of propostional formulas, in which each variable is quantified by universal or existential quantifiers. 
QBF does not increase the expressive power of propositional logic, it simply offers an exponentially more succinct encoding of problems. As a result, problems from various important industrial fields, such as formal verification and model checking can be encoded succintly as QBFs. This leads to the desire of building efficient QBF-SAT solvers. Let us define now QBFs more formally.

\noindent
{\bf Quantified Boolean Formulas (QBFs)}: It extend propositional logic with
Boolean quantifiers with the standard semantics that $\forall x. F$ is
satisfied by the same truth assignments as $F|_{x = 0} \wedge F|_{x =
	1}$ and $\exists x. F$ as $F|_{x = 0} \vee F|_{x = 1}$. 

We say that a QBF is in \emph{closed prenex form} with a CNF (Conjunctive Normal Form) matrix, if the QBF instance is of the form $\Pi.\phi$, where $\Pi = \mathcal{Q}_1 X_1 \mathcal{Q}_2 X_2 \dots \mathcal{Q}_n X_n$ is called the quantifier prefix, with $\mathcal{Q}_i \in \{\exists, \forall\}$, and $\phi$ is a quantifier-free 
CNF formula in the variables $X_1 \cup \ldots \cup X_n$. We have $X_i \cap X_j = \emptyset$, $\mathcal{Q}_i \neq \mathcal{Q}_{i+1}$, and $\mathcal{Q}_1, \mathcal{Q}_n = \exists$. 


A literal is either a variable or its complement. For a clause $C$, $var(C)$ is a set containing all the variables of $C$. For a QBF $\mathcal{F}$, $var(\mathcal{F})$ is a set containing all the variables of $\mathcal{F}$ (i.e., $var(\mathcal{F}) = \cup_{C \in \mathcal{F}} \{var(C)\}$). A quantifier $Q(\Pi,\ell)$ of a literal $\ell$ is $\mathcal{Q}_i$ if the variable $var(\ell)$ of $\ell$ is in $X_i$. A literal $\ell$ is existential if $Q(\Pi,\ell) = \exists$ and universal if $Q(\Pi,\ell) = \forall$. For literals $\ell$ and $k$, with $Q(\Pi,\ell) = \mathcal{Q}_i$ and $Q(\Pi,k) = \mathcal{Q}_j$, we say that the literal $\ell$ is on the left of literal $k$ ($\ell \leq_{\Pi} k$) if and only if $i \leq j$. 
Within a clause, we order literals according to their ordering in the quantifier prefix. We assume that QBF is in this form, unless noted otherwise. The order $\leq_{\Pi}$ is arbitrary extended to variables withing each block $X_i$. For a literal $\ell$, $level(\Pi,\ell) = i$ if $var(\ell) \in X_i$.



A QBF \QBF{\Q_1 x_1 \cdots \Q_k x_k}{\phi} can be seen as a game between two players: {\em universal} $(\forall)$ and {\em existential} $(\exists)$.
In the $i^{th}$ step of the game, the player ${\cal Q}_i$ assigns a value to the variable $x_i$. 
The existential player wins if $\phi$ evaluates to $1$ under the assignment constructed in the game. 
The universal player wins if $\phi$ evaluates to $0$.

An assignement tree~\cite{QRATOriginal17} $T$ of a QBF $\mathcal{F}$ is a complete binary tree of depth $|var(\mathcal{F})| + 1$. Internal nodes corresponds to the variables of $\mathcal{F}$. The order of the variables in $T$ respect the orders in the prefix in $\mathcal{F}$. Internal nodes which corresponds to an existential variables act as an OR-nodes, whereas universal nodes corresponds to an AND-nodes. Every internal nodes $x$ has two children, one for literal $x$ (true) and one for literal $\neg x$ (false). An existential node is labelled with $\top$ (true) if at least one of its children is labelled with $\top$ (true). A universal node is labelled with $\top$ if both of its children are labelled with $\top$. 
A path in $T$ is a sequence of literals. A path $\tau$ from root to a leave in $T$ is a complete assignement and the leave is labelled with the value of QBF under $\tau$. The QBF $\mathcal{F}$ is true if the root of $T$ is labelled with $\top$, and $\mathcal{F}$ is false if the root is labelled with $\perp$.

A subtree $T'$ of the asignment tree $T$ is a pre-model~\cite{QRATOriginal17} of the QBF $\mathcal{F}$ if the root of $T$ is also the root of $T'$, both children of every universal node must be present in $T'$, and exactly one child of every existential node are in $T'$. We say that a pre-model $T'$ of a QBF $\mathcal{F}$ is a model of $\mathcal{F}$, denoted $T' \models \mathcal{F}$ if every node in $T'$ is labelled with $\top$. Similarly, a false QBF has countermodels which can be defined dually.

A QBF $\mathcal{F}$ is satisfiable if it has at least one model. For QBFs $\mathcal{F}$ and $\mathcal{H}$, we say that $\mathcal{H}$ is implied by $\mathcal{F}$, denoted $\mathcal{F} \models \mathcal{H}$ if, for all $T$, if $T \models \mathcal{F}$ then $T \models \mathcal{H}$.

We say that two QBFs $\mathcal{F}$ and $\mathcal{H}$ are (logically) equivalent denoted $(\mathcal{F} \equiv \mathcal{H})$ iff  $\mathcal{F} \models \mathcal{H}$ and $\mathcal{H} \models \mathcal{F}$, and satisifability equivalent, denoted $\mathcal{F} \equiv_{sat} \mathcal{H}$, iff $\mathcal{F}$ is satisifable whenever $\mathcal{H}$ is satisfiable.

\subsection{QBF proof complexity}
Since past few decades, QBFs solving is an active area of research. As a result, several proof systems for QBFs have been developed. For example, Kleine B{\"u}ning~et~al. in~\cite{BuningKF95} introduced {\bf Q-resolution} (\qrc), which is an extension of the propositional resolution proof system for QBFs. Before defining \qrc, let us quickly define the Resolution (\res) proof system.

{\bf Resolution}: ({\res}) is well studied propositional proof system introduced by Blake in \cite{Bla37} and proposed by Robinson in \cite{Rob65} as automated theorem proving. 
The lines in the resolution proofs are clauses.
Given a CNF formula $F$, {\res} can infer new clauses according to the following inference (resolution) rule : 
$$\frac{C \vee x \hspace{5mm} D \vee \neg x}{C \vee D},$$
where $C$ and $D$ are clauses and $x$ is a variable being resolved, called as pivot variable. The clause $C \vee D$ is called the resolvent. Let $F$ be an unsatisfiable CNF formula. A resolution proof (refutation) $\pi$ of $F$ is a sequence of clauses $D_1, \dots, D_l$ with $D_l = \Box$ and each clause in the sequence is either from $F$ or is derived from previous clauses in the sequence using the above resolution rule.

We can also view $\pi$ as a directed acyclic graph $G_{\pi}$, where the source nodes are the clauses from $F$, internal nodes are the derived clauses and the empty node is the unique sink. Edges in $G_{\pi}$ are from the hypotheses  to the conclusion for each resolution step. 
In $G_{\pi}$, we say that a clause $C$ is descendant to a clause $D$ if there is a directed path from $C$ to $D$.

{\bf Q-resolution}: \qrc uses the resolution rule, but with a side condition that the pivot variable must be existential, and the resolvent does not have simultaneously a $z$ and $\neg z$. In addition, \qrc has a universal reduction rule ($\forall$-Red) which allows dropping a universal literal $u$ from a clause provided the clause has no existential literal $\ell$ with $u \leq_{\Pi} \ell$. That is, there exists no existential literal to the right of the reduced literal $u$. We say in this case, that $u$ is not {\bf blocked}.  

{\bf QU-resolution}~\cite{Gelder12}, removes the restriction from {\qrc} that the resolved variable must be existential variable and allows resolution on universal variables as well. Thus QU-resolution ({\qurc}) is classical resolution augmented with a $\forall$-Red rule. Interested readers are referred to~\cite{Shukla17}, for more details on QBF proof complexity.

\subsection{QBFs solvers}\label{subsec:qbf-solver}
Since, the propositional satisfiablity (SAT) problem is NP-complete, and the QBF-SAT problem is PSPACE-complete~\cite{AroraBarak09}, only algorithms with exponential worst-case time complexity are known for these problems. Despite of this, several efficient SAT solvers, capable of solving the satisfiability instances with thousands of variables and clauses, have been designed. For example, the GRASP search algorithm for the satisfiability problem~\cite{Marques-SilvaS99}, implemented by Marques-Silva and Sakallah. GRASP is based on conflict-driven clause learning (CDCL) algorithms~\cite{Ashish2005}.

Motivated by the success of propositional SAT solvers, many QBF solvers have also been developed. We would like to mention some important ones:

{\bf DepQBF}: DepQBF, developed by Lonsing and Egly~\cite{LonsingE17}, is based on QCDCL~\cite{ZhangM02}, which is an extension of CDCL for QBFs. DepQBF, implements a variant of QCDCL, which is based on a generalization of \qrc. This generalization is due to some additional axioms. Q-Res proofs can be extracted from instances of the QCDCL algorithm. This generalization, in DepQBF implementation, helps to produce exponential shorter refutation, than the traditional QCDCL-based solvers. 

{\bf RAReQS}: Recursive Abstraction Refinement QBF Solver~\cite{JanotaKMC12}, is a recursive algorithm for QBF solving, based on the Counter-example Guided Abstraction Refinement (CEGAR) technique~\cite{ClarkeGJLV03}. Initially, CEGAR was used to solve QBFs with two level of quantifiers~\cite{JanotaS11}. RAReQS, extends the technique to QBFs with an arbitrary number of quantifiers using recursion. 

CEGAR technique was designed to tackle problems whose implicit representation is infeasible to solve, and an abstract instance is tackled instead. Thus, in CEGAR-based algorithm, we have a concrete (an original) representation, and an abstract representation of a problem. Both the representations are connected as follows: if the abstract problem has no solution, then concrete problem also does not have a solution. If a solution to the abstract problem is not a solution to the concrete problem, then a counterexample can be constructed for this fact. And finally, if there are no counterexamples for a solution to the abstract problem, then the solution, is also a solution to the concrete problem. 

Based on the above properties one can compute a solution to the original problem as follows: compute a solution to the abstract problem, if no such solution exists, return: no solution to the original problem. Otherwise, check if it is also a solution to the original problem: if no counterexample exists, return: the same solution is a solution to the original problem. Otherwise, use the counterexample to refine the abstract representation and repeat. This gives the name: CounterExample Guided Abstract Refinement technique.

{\bf GhostQ}: GhostQ solver~\cite{KlieberSGC10} is originally implemented to support the QBF solving for instances which are not in prenex from. GhostQ is a DPLL-based QBF solver which uses ghost variables. Ghost literal, first introduced in~\cite{GoultiaevaIB09}, is a powerful propogation technique for QBFs. Interested readers are referred to~\cite{KlieberSGC10}.

\subsection{First-order predicate logic}\label{subsec:first-order}

First-order predicate logic extends the propositiona logic by adding both the predicates and the quantification.
A Predicate $P(x_1,x_2,\dots,x_n)$ is an $n$-ary function whose range is only $0$ (false) or $1$ (true). Unlike propositiona logic, where variable can take only $0$-$1$ values, in predicate logic variables can take any values from some universal set $U$. Predicates with arity $0$ are just the propositional variables. Some examples of predicates are as follows. Assume the set of all humans as the universal domain set $U$ for the below examples.  \\
\begin{example}
``$x$ is a male $\equiv$ male($x$)": here variable $x$ ranges over the universal set $U$. The predicate male($x$) is true iff $x$ is male.
\end{example}
\begin{example}
 ``$x$ is the mother of $y$ $\equiv$ mother($x$, $y$)": again $x,y$ ranges over the universal set $U$, and the predicate mother($x, y$) is true iff the variables $x$ and $y$ are initialized with some values $a, b \in U$ respectively, such that $a$ is the mother of $b$.
\end{example}

Clearly, as compared to the propositional logic, in the first-order predicate logic, it is possible to express knowledge much easier, for example the statement ``for every $x$, if $x$ is an indian then $x$ is a human" can be expressed as: 
$$\forall x~~ indian(x) \rightarrow human(x).$$
As another example, consider the following statement: If $x$ is a female, and parents of $x$ and $y$ are same, then $x$ is a sister of $y$. This statement can be expressed as: $$\forall x, y, m, f~~Female(x) \wedge Parent(x,m,f) \wedge Parent(y, m, f) \rightarrow Sister(x,y).$$ 

Obseve that the above example is equivalent to the following: 
$$\forall x, y, m, f~\Big(\neg \text{Female}(x) \vee \neg \text{Parent}(x,m,f) \vee \neg \text{Parent}(y, m, f) \vee \text{Sister}(x,y)\Big).$$ 

This rule has just one positive predicate, and hence is a Horn clause. A {\bf first-order predicate Horn logic} consists of Horn clauses only. However, each literal on a clause can be replaced by an arbitrary predicates of any arity. Also, the variables can be quantified arbitrary.

We skip the definition of first-order predicate logic more formally here. Interested readers are referred to the book by Robert S. Wolf~\cite{robert2005}.\\

\noindent
{\bf Substitution and Unification}: Given a first-order predicate formula $F(x_1,\dots,x_n)$, over distinct variables, $x_1,\dots,x_n,$ a substitution is a process of binding each variable $x_i$ with a predicate $P_i$ (denoted, $x_i/P_i$). The arity of $P_i$ is allowed to be zero. 

Recall, in \res, the resolution rules are performed over clauses with complementary literals $x$ and $\neg x$. 
Consider two literals $P(a)$ and $\neg P(x)$ appearing in two distinct clauses in a first-order predicate logic. The literals are almost complementary: the first contains a constant, whereas the second contains a variables. The subsitution can be applied here two make the literals complementary: substitute $x/a$ in the second. Thus, the literal $P(a)$ and $\neg P(x)$ becomes complementary after the substitution. One may now perform resolution step on the same. Thus substitution can be applied to make two first-order predicate formulas syntactically equal. 

The {\bf unification algorithm} is a general method for comparing first-order predicate formulas. The algorithms also computes the substitution which is needed to make the formulas syntatically equal. Interested readers are referred to~\cite{unification91}.

\subsection{Prolog}
	Prolog stands for PROgramming in LOGic. Prolog is a declarative programming language. Unlike, the procedural  programming languages, where we have to tell the computer what to do when, and how to achieve a certain goal, in Prolog on the other hand, we only need to tell what is true, and then quering it to draw conclusions. To be precise, in Prolog the description of a problem and the procedure for solving it are separated from each other. According to Robert Kowalski, a Prolog algorithm can be expressed as:
$$\text{algorithm = logic }+ \text{ control}$$ 
In the above equation, the logic part gives the description of the problem, that is, what the algorithm should do, and the control part indicates how it should be done. In particular, the program logic is expressed as facts and rules, and a computation is initiated by running a query over these facts and rules. 

The facts and rules are represented as {\bf Horn formulas}. Horn formulas are a special case of CNF formulas, where each clause is allowed to have at most one positive literals. 

Once a query as a goal clause is provided, Prolog tries to find a resolution refutation of the negated query. If the refutation is found, the query is said to be sucessful. This part constitute the control part. 

To be precise, the Prolog control part is based on SLD-resolution, which is an important refinement of Linear Resolution~\cite{ColmerauerR93}. A linear resolution derivation of a clause $C$ from a CNF formula $F$ is a sequence of clauses $C_1, C_2, \dots, C_m$, such that $C_1 \in F$, $C_m = C$, and for every $i<m$, $C_{i+1}$ is the resolvent of $C_i$ either with a clause $D \in F$ or with a clause $C_k$  for $k<i$. In order to define the basic principle used in processing a logic program in Prolog, we need to first define the SLD-resolution (Definition~\ref{def:sld-resolution}).  

\subsection{Proof system for Prolog: \sldres}\label{subsec:proofsystem}
 In order to cut down the search space for SAT solvers, several refinements of resolution proof system have been introduced. One of the most important refinement is the linear resolution. Linear resolution is known to be complete and sound for unsatisfiable CNF formulas \cite{Johannsen05}. To even cut down the search space, several refinements of linear resolution have been proposed, for example, \slin~\cite{Loveland70, Zamov71}, \tlin~\cite{Kowalski71}, \slres~\cite{Kowalski71}, and \sldres~\cite{Kowalski74}. All the above mentioned refinements are known to be complete and sound for unsatisfiable CNF formulas, except the \sldres. As Prolog is based on \sldres we define it next.
	
\subsubsection{\sldres}\label{subsec:sld}
\sldres, a refinement of linear resolution, was introduced by Robert Kowalski~\cite{Kowalski74}. The name SLD-Resolution stands for SL resolution with definite clauses. \sldres is complete and sound for propositional Horn formulas. Before introducing \sldres, let us quickly re-visit propositional Horn formulas in detail.

\subsubsection{Propositional Horn Formulas}\label{subsec:propHorn}
Horn formulas are a special case of CNF formulas, named after the logician Alfred Horn. A clause $C$ is a Horn clause if and only if $C$ contains at most $1$ positive literal. A Horn formula is a conjunction of Horn clauses. A definite Horn clause is a clause with exactly one positive literal and zero or more negative literals. A conjunction of definite Horn clauses is called a definite Horn formula. We usually denote a definite Horn clause $(x \vee \neg x_1 \vee \dots \vee \neg x_k)$ as $x \leftarrow x_1, \dots, x_k$, which is equivalent to the following:
$$\text{IF } x_1 \wedge x_2 \wedge \dots \wedge x_k \text{ then } x$$
Here, $x, x_1, \dots, x_k$ are variables. The variable $x$ is called the head of the clause, and $x_1,\dots,x_k$ is called the body of the clause. We denote the head of a Horn clause $C$ by $C^+$ (or $head(C)$) and the body by $C^-$. A clause with only negative literals are referred as negative clause. Clearly a negative clause is also a Horn clause. We usually refer a negative Horn clause as a {\bf goal} clause. We denote a goal clause $(\neg x_1 \vee \neg x_2 \vee \dots \vee \neg x_k)$ as $ \leftarrow x_1, x_2, \dots, x_k$. The empty clause $\Box$ is also considered as a goal clause. We quickly list some important facts about Horn formulas.

\begin{lemma}\cite{Bunning1999}\label{lem:Horn-properties}
For any Horn formula (other than the empty clause), the following holds:
\begin{enumerate}
	\item A definite Horn formula is satisfiable.

	\item A Horn formula $F$ is satisfiable if $F$ contains no positive unit clause (clauses with only one literal are called unit clauses). 
\end{enumerate}
\end{lemma}

\begin{proof}
Every clause in a definite Horn formula has exactly one positive literal. Thus the assignment which assigns each variable the value true, satisifies the formula. 

If the Horn formula has no positive unit clause, then each clause contains at least one negative literals. Again the assignment which assigns each variable the value false, satisfies the formula.
\end{proof}
Therefore, for a Horn formula to be unsatisfiable, it must contain at least one negative clause. In fact the following also holds:

\begin{lemma}\cite{Bunning1999}
Let $F = F_1 \cup F_2$ be an unsatisfiable Horn formula, where $F_1$ contains only definite Horn clauses and $F_2$ contains only negative clauses. Then 
$$F \text{ is unsatisfiable } \iff \exists C \in F_2: F_1 \cup \{C\} \text{ is unsatisfiable}$$

\end{lemma}

\begin{proof}
Suppose not. That is, let $F$ be unsatisfiable and for all $C \in F_2$, $F_1 \cup \{C\}$ is satisfiable. Assuming this we show that $F$ is satisfiable as well, which is clearly a contradiction. 

Let variables $x_1, x_2, \dots x_m$ are consequences of $F_1$, that is $F_1 \implies x_i$ for $i \in [m]$. Clearly, $var(C) \nsubseteq \{x_1,\dots,x_m\}$ for every $C \in F_2$, otherwise $F$ would have been satisfiable: $F_1$ is satisfiable from Lemma~\ref{lem:Horn-properties}, and if $var(C) \in \{x_1\dots,x_m\}$ for all $C$ in $F_2$, then $F_2$ is also satisfiable.  

Let $\mathcal{I}$ be the following assignment: $\mathcal{I}(x_j) = 1$ for $j\in[m]$, and $\mathcal{I}(y) = 0$ for all other variables $y$. Clearly, $\mathcal{I}(C) =1$ for all $C\in F_2$: as $C$ is negative clause and at least one variable $y$ must belong to $C$. Therefore we have $\mathcal{I}(F_2) = 1$.

On the other hand, observe that $\mathcal{I}(F_1) = 1$ as well: since $\mathcal{I}(x_j) = 1$ and $F_1 \implies x_j$, for $j\in [m]$, therefore $\mathcal{I}(F_1)$ must be $1$.


\end{proof}

Since we are discussing Horn formulas, we quickly introduce an important refinement of resolution proof system which is incomplete for CNF formulas but complete for Horn formulas: unit resolution (\unit). A resolution proof system is called a unit resolution proof systems iff for every resolution step one of its parent clause is a unit-clause. We have the following:

\begin{lemma}\cite{Bunning1999}
Let $F$ be a Horn formula, then
$$F \text{ is unsatisfiable } \iff F \derivesunit ~\Box$$
\end{lemma}

\begin{proof}
Since \res simulates \unit and \res is sound, we only need to prove that \unit is complete. We prove this by induction on the number of variables in $F$. For $n=1$ clearly $F \equiv x \wedge \neg x$, and the Lemma follows. 

Assume that $n>1$. As $F$ is unsatisfiable Horn formula, it must have a positive unit clause $x$. Resolve $x$ with all the clauses having $\neg x$. Clearly the resulting formula $F'$ is indeed Horn. Also observe that $F'$ is unsatisfiable (the proof is exactly as in the completeness proof of resolution). Clearly, $F'$ is an unsatisfiable Horn formula over $n-1$ variables. The refutation follows from the induction hypothesis.
\end{proof}

   \begin{definition}[SLD-resolution]\cite{Gallier15} \label{def:sld-resolution}
Let $F$ be a Horn formula. Let $D, G \subseteq F$ be disjoint subsets of clauses from $F$ such that $D$ is the set of all definite Horn clauses, and $G = \{ G_1, \dots, G_q\}$ be the set of all goal (negative) clauses. Let $C$ be a goal clause. An \sld derivation of $C$ from $F$ is a sequence of negative clauses $\pi = N_0, N_1, \dots, N_m$ such that:
\begin{enumerate}
	\item $N_0 = G_j$, where $G_j$ is some goal from $F$.
	\item 
Each $N_i$ is a resolvent of $N_{i-1}$ and a definite input clause $C_i \in D$, resolved over the head $x$ of $C_i$. The variable $x$ (i.e., the literal $\neg x$) of $N_{i-1}$ is the selected variable in the body of $N_{i-1}$. We call, the clause $N_{i-1}$ a {\bf centre} clause, and the clause $C_i$, a {\bf side} clause.
	\item $N_m = C$. 
\end{enumerate}
$N_0$ is the top clause, and the $C_i$ are the input clauses of this \sld derivation. If an \sld derivation of $C$ from $F$ exists, we write $F ~\derivessld~C$. If $N_m = \Box$, we call $\pi$ an \sld refutation of the Horn formula $F$. 
	
\end{definition}

Observe that, an \sld refutation is not just a linear resolution but also an input and a negative resolution (every resolution step has at least one negative clause as parent). In fact, all the clauses $N_i$ is the \sld derivation are goal clauses. This is because the top clause $N_0$ is a goal clause. If we pick $N_0$ to be a definite clause from $F$, then all clauses in the \sld derivation will be definite clauses. Note that the input clauses of any \sld derivations are always the definite clauses.


We know that \sld proof system is sound and complete for Horn formulas. 

\begin{theorem}\cite{Ronald97}
\sld proof system is sound and complete for Horn formulas.
\end{theorem}

Prolog treats clauses as multisets of literals (it retains all the repeated literals in a resolvent). For Prolog not only the order of literals within a clause matters, but also the order of clauses within a Horn formula matters. The Horn formulas in which both the clauses and the literals within the clauses are given a fixed ordering are called a {\bf Prolog program}. Given a Prolog program $\pi$ and a query $\gamma$ one can easily define the Prolog mechanism via a recursive algorithm. Interested readers are referred to~\cite{Bunning1999}.

\section{SLD-resolution for QBFs (\sldQ)}\label{sec:sldQ}
	If we allow universal variables, along with existential variables, to be used as a pivot variable, for the resolution step (as in \qurc), then \sldres can be easily extended for QBFs: just add the $\forall$-Red rules as in~\cite{BeyersdorffBC16}. This proof system is interesting, and we plan to study this in future. However, as per our knowledge, there exists no QBF solver based on \qurc. Therefore, extending \sldres to QBFs with a restriction that only existential variables are allowed, as a pivot variable in the resolution step (as in Q-resolution), is more important. We work on this direction now. As, \sldres proof system is complete only for Horn formulas, let us define quantified Horn formulas first:
	\begin{definition}\label{def:QHorn}\cite{Bunning1999}
		A quantified Boolean formula $\mathcal{F} \equiv \mathcal{Q}_1 X_1 \mathcal{Q}_2 X_2 \dots \mathcal{Q}_nX_n (C_1 \wedge \dots \wedge C_m)$ is a quantified Horn formula iff each clause $C_i$ in the matrix has at most one positive literals (existential or universal).
		We use notations as in the propositional case.
	\end{definition}

It is well known that the satisfiability problem for the quantified Horn formula can be solvable in polynomial time~\cite{Bunning1999}, however the equivalence problem for them are coNP complete~\cite{Bunning1999}. 
 
	It is clear that every clause in any quantified Horn formula $\mathcal{F}$ belongs to one of the following set:
	
	\begin{itemize}
		\item $F_{\exists}$: set of all clauses $C$ with exactly one positive existential literal $\ell$. 
		
		\item $F_{\forall}$: set of all clauses with exactly one positive universal literal.
		
		\item $F_{\text{goal}}$: set of all goal clauses with only negative (existential or universal) literals.
	\end{itemize}
	\begin{observation}\cite{Bunning1999} \label{obs:1}
		A quantified Horn formula $\mathcal{F}$ is false if and only if there exists a clause
		$C' \in F_{\forall} \cup F_{\text{goal}}$ such that $\mathcal{Q}_1 X_1 \mathcal{Q}_2 X_2 \dots \mathcal{Q}_nX_n. (F_{\exists} \wedge C')$ is false.
	\end{observation} 

	\begin{observation}\label{obs:2}
		Let $\mathcal{F}$ be a false quantified Horn formula. Let us assume, that in Observation~\ref{obs:1}, we have $C' \in F_{\text{goal}}$. That is, assume $C' \in F_{\text{goal}}$, and $\mathcal{Q}_1 X_1 \mathcal{Q}_2 X_2 \dots \mathcal{Q}_nX_n. (F_{\exists} \wedge C')$ is false. Then for every negated existential literal $q$ in $C'$ (i.e, $\neg q \in C'$), we have a clause $C \in F_{\exists}$, with the variable $q$ (i.e., literal $q$) as its head.		
	\end{observation} 

	\begin{proof}
		We know that the initial QBF $\mathcal{F}$ is false. Therefore, there does not exists any winning strategy for the existential player.

Consider the following strategy of the existential player: assign $1$ to all the heads of the clauses from $F_{\exists}$. Let us call this assignment $\alpha$. For sure, $\alpha$ satisfies all the clauses of $F_{\exists}$. However, as the original formula is false, this assignment does not satisfy the clause $C'$. In fact $C'|{\alpha}$ is either $0$, or is left with a bunch of negated universal literals. Otherwise, $\alpha$ can be extended as a winning strategy for the existential player. 
		
		This tells us that, every negated existential literal $(\neg \ell)$ of $C'$ has a corresponding head $\ell$ in some clause of $F_{\exists}$. If not, one can extend $\alpha$ to winning startegy for the existential player.
	\end{proof}

	\begin{observation}\label{obs:3}
		Let us assume, that in Observation~\ref{obs:1}, we have $C' \in F_\forall$. Then, for every negated existential literal $q$ in $C'$ (i.e, $\neg q \in C'$), 
we have a clause $C \in F_{\exists}$, with the variable $q$ (i.e., literal $q$) as its head.		
	\end{observation} 
	
	\begin{proof}
		The proof is exactly as in Observation~\ref{obs:2}.
	\end{proof}

	Let us define \sldres for QBFs (\sldQ). An \sldQ derivation of a clause $C$ from a quantified Horn formula $\mathcal{F}$ is
exactly as in the propositional case, except that a clause $N_i$ in the sequence may
also be derived via a $\forall$-Red rule from $N_{i-1}$. That is, by deleting a universal variable
of $N_{i-1}$ which has not been blocked in $N_{i-1}$. To be precise,

	\begin{definition}[\sldQ]\label{def:sldq-res}
Let $\mathcal{F}$ be a false quantified Horn formula. Let $C$ be a goal clause. An
\sldQ derivation of $C$ from $\mathcal{F}$ is a sequence of negative clauses $\pi = N_0, N_1, \dots, N_m$
such that:

\begin{enumerate}
	\item $N_0 \in \mathcal{F}$, is an initial goal clause.

	\item Each $N_i$ is a resolvent of $N_{i-1}$ and a definite input clause $C_i$, resolved over the
head $x$ of $C_i$. The variable $x$ (i.e., the literal $\neg x$) of $N_{i-1}$ is the selected variable
in the body of $N_{i-1}$. (Observe that $x$ is an existential literal by definition).
Resolution is only allowed on existential variables.

	\item Each $N_i$ is derived from $N_{i-1}$ via a $\forall$-Red step. That is, by deleting a universal
variable $u$ from $N_{i-1}$ such that $u$ is not blocked in $N_{i-1}$.

	\item $N_m = C$.

\end{enumerate}
If $N_m = \Box$ then $\pi$ is an \sldQ refutation of the quantified Horn formula $\mathcal{F}$. We denote by $\mathcal{F}~\derivessldQ~\Box$ the fact that $\mathcal{F}$ has an \sldQ refutation. 
	\end{definition} 

	\begin{theorem}
		\sldQ is sound and complete for false quantified Horn formulas.
	\end{theorem}
	\begin{proof}
		As \qrc is sound and can simulate \sldQ, we only need to prove completeness. Let $\mathcal{F}$ is a false quantified Horn formula. Then, by Observation~\ref{obs:1}, $\exists C' \in F_{\text{goal}} \cup F_{\forall}$ such that $\mathcal{Q}_1 X_1 \mathcal{Q}_2 X_2 \dots \mathcal{Q}_nX_n (F_{\exists} \wedge C')$ is false. So we have two cases:\\
		{\bf Case 1}: When $C' \in F_{\text{goal}}$: If $C'$ has only negated universal literals, just apply the $\forall$-Red rules and derive the empty clause. Clearly, this is an \sldQ refutation. Otherwise, let $C'$ has the following existential negated literals (with some negated universal literals as well): $\neg x_1, \neg x_2, \dots, \neg x_k$. Let $C_i \in F_{\exists}$ be the clause with head $x_i$. By Observation~\ref{obs:2}, we must have a $C_i \in F_{\exists}$ for every $x_i$.

We start constructing an \sldQ refutation with the top clause $C'$ as follow: resolve $C'$ and the clause $C_1$ with pivot $x_1$. Let the resolvent is $C'_1$. Resolve $C'_1$ with the clause $C_2$ with pivot $x_2$. and so on. Call the resolvent after $k$ step as $C'_k$. 

We {\bf claim} that $C'_k$ is either an empty clause, or consists of only negated universal literal. If this claim is true, we are done. We have an \sldQ refutation of $\mathcal{F}$.  

Suppose not. Then, $C'_k$ has an existential variable say $q$. Clearly, $q$ is a negated literal. Because, each $C'_i$ is in fact a negated clause only. Observe that $\neg q \notin C'$, otherwise, it would have been resolved via an input clause containing $q$ as head. So, it must be the case, that $\neg q$ belongs to some $C_{i-1} \in F_{\exists}$ and introduced in the resolvent $C'_{i}$. 

Assuming this, we now give a winning strategy for the existential player for the QBF $\mathcal{F}$: assign $1$ to all the heads of the clauses from $F_{\exists}$, except the clause $C_{i-1}$. In $C_{i-1}$ assign $0$ to the head $x_{i-1}$ and a $0$ to the negated literal $ q$. Clearly, the assignement satisfies all clauses from $F_{\exists}$. Also, it satisfies the clause $C'$: as $\neg x_{i-1} \in C'$ and $x_{i-1}=0$. Note that $\neg x_{i-1}$ surely is present in $C'$, that is why we resolved $C'_{i-1}$ with $C_{i-1}$.  This proves our claim.

{\bf Note}: as $\mathcal{F}$ is a QBF Horn formula, we do not have any universal variables $u$ and $\neg u$ together in any resolvent $C'_i$.\\
		{\bf Case 2}: When $C' \in F_{\forall}$: This case is similar to the Case 1. Again, begin with clause $C'$ and start resolving the negated existential literals of $C'$ with the corresponding definite clauses. This is possible by Observation~\ref{obs:2}. At the end, we are left with the empty clause or with only negated universal literals.
	\end{proof}

\section{Extending Prolog to QBFs (QBF-Prolog)}\label{sec:qbf-prolog}
In this Section, we develop theory for the QBF-Prolog. We follow the ideas from~\cite{Bunning1999,QRATPlus18}.
Given facts and rules as a quantified Horn formulas $\mathcal{F} = \Pi.~\phi$, and a query Horn clause $C \equiv x \leftarrow x_1, x_2,\dots,x_n$. We need to answer whether 
$$\Pi.\phi \models \Pi.(\phi \wedge \{C\})?$$
We note that $C$ may have new variables as well, and there are no restrictions on how these variables are quantified and where they are put within the prefix. 

Let us first develop theory for the case when the query clause $C$ has no new variables. 
We show that the other case is trivial. That is, we show that the case in which the query clause contains new literals have a straightforward answer: if the new literal in the query clause is existential, then the answer is yes. If the new literal $u$ is universal, then one may drop $u$ from the query clause without effecting the answer.  

\subsection{Theory for QBF-Prolog when the query Horn clause has no new variables}\label{sec:query-no-new-variable}
Let $\mathcal{F} = \Pi.\phi$ be a quantified Boolean Horn formula. Let $C$ be a Horn clause, such that $var(C) \subseteq var(\mathcal{F})$. We need to answer whether
$$\Pi.\phi \models \Pi.(\phi \wedge \{C\})?$$ 

We use the ideas from~\cite{QRATPlus18}.\\

\noindent
{\bf Abstractions}: Given a QBF $\mathcal{F} = \Pi.\phi$ with prefix $ \Pi = \mathcal{Q}_1 X_1 \dots \mathcal{Q}_iX_i\mathcal{Q}_{i+1}X_{i+1}\dots\mathcal{Q}_nX_n$, and an $i$ with $0 \leq i \leq n$. An abstraction QBF is defined as $\mathcal{F}_i = \Pi_i.\phi$ where 
$$\Pi_i = \exists (X_1 \cup \dots \cup X_i) \mathcal{Q}_{i+1}X_{i+1}\dots \mathcal{Q}_nX_n$$
Note that $\mathcal{F}_0 \equiv \mathcal{F}$ and $\mathcal{F}_n$ is just the CNF formula $\phi$ over all existential variables $X_1\cup\dots \cup X_n$. We state the following Lemma from~\cite{QRATPlus18} regarding abstractions without proof:

\begin{lemma}[~\cite{QRATPlus18}]\label{lemma-abstraction}
Let $\mathcal{F} = \Pi.\phi$ and $\mathcal{F}' = \Pi.\phi'$ be two QBFs with same prefix. Then for all $i$, if $\Pi_i.\phi \equiv \Pi_i.\phi'$ then $\mathcal{F} \equiv \mathcal{F}'$.
\end{lemma}

Let $i = max_{\ell \in C}\{level(\Pi,\ell)\}$, that is, $i$ be the maximum level of literals in the query Horn clause $C$. For a clause $C$ (say, $(\ell_1 \vee \dots \vee \ell_k)$), we define $\overline{C}$ as the complement of $C$, that is, conjunction of all the negated literals of $C$ (that is, $(\neg \ell_1 \wedge \dots \wedge \neg \ell_k)$). We have the following:

\begin{lemma}[\cite{QRATPlus18}]\label{lemma:no-new-variable}
Let $\mathcal{F} = \Pi.\phi$ be a Horn QBF, $C$ be a Horn clause with $var(C) \subseteq var(\mathcal{F})$ and $i = max_{\ell \in C}\{level(\Pi,\ell)\}$. If $\Pi_i.(\phi \wedge \overline{C})~ \derivessldQ ~\Box $ then $\Pi.\phi \equiv \Pi.(\phi \wedge C)$
\end{lemma}
Note that in~\cite{QRATPlus18}, they proved the result via quantified unit propagation~\cite[Definition 9]{QRATPlus18}. However, here we show that the result is valid for \sldQ as well.

Lemma~\ref{lemma:no-new-variable} follows from Lemma~\ref{lemma-abstraction}, and Lemma~\ref{lemma2} below.

\begin{lemma}[\cite{QRATPlus18}]\label{lemma-base-case}
Let $\Pi.\phi$ be a Horn QBF, with $\Pi =  \mathcal{Q}_1 X_1\dots\mathcal{Q}_nX_n$ and $C$ be a Horn query clause with $var(C) \subseteq X_1$, and $i = max_{\ell \in C}\{level(\Pi,\ell)\} = 1$. If $\Pi_1.(\phi \wedge \overline{C})~ \derivessldQ ~\Box $ then $\Pi_1.\phi \equiv \Pi_1.(\phi \wedge C)$. 

As $\Pi_1.\phi$ is equal to $\Pi.\phi$. We have that if $\Pi.(\phi \wedge \overline{C})~ \derivessldQ ~\Box $ then $\Pi.\phi \equiv \Pi.(\phi \wedge C)$.
\end{lemma}

\begin{proof}
Suppose not, and there exists a $T$ (assignment tree) with $T \models \Pi.\phi$ but $T \not \models \Pi.(\phi \wedge C)$. It follows that there exists a path $\tau$ in $T$ with $\tau(C) = \perp$. As $\Pi_1.(\phi \wedge \overline{C})~\derivessldQ~\Box$, the Horn QBF $\Pi.(\phi \wedge \overline{C})$ is unsatisfiable, and so $T \not \models \Pi.(\phi \wedge \overline{C})$. Since $\tau(C) = \perp$, we have $\tau(\overline{C}) = \top$ and hence $T \models \Pi.(\phi \wedge \overline{C})$, a contradiction.
\end{proof}

\begin{lemma}[\cite{QRATPlus18}]\label{lemma2}
Let $\Pi.\phi$ be a Horn QBF, $C$ be a Horn query clause and $i = max_{\ell \in C}\{level(\Pi,\ell)\}$. If $\Pi_i.(\phi \wedge \overline{C})~ \derivessldQ ~\Box $ then $\Pi_i.\phi \equiv \Pi_i.(\phi \wedge C)$.
\end{lemma}

\begin{proof}
Since all variables from $C$ are existentially quantified in the abstraction $\Pi_i.(\phi \wedge \overline{C})$, they all belongs to the first quantifier block and the Lemma follows from Lemma~\ref{lemma-base-case}.
\end{proof}

Using the ideas from~\cite{Bunning1999} and Lemma~\ref{lemma:no-new-variable}, we next develop an efficient implementation model. We start by presenting an exponential time algorithm for the same.

\subsection{Exponential time algorithm for QBF-Prolog when the query clause has no new variables}
In this Section we develop an efficient implementation for the QBF-Prolog when the query clause has no new variables. In particular, we need to develop an efficient implementation details for the Lemma~\ref{lemma:no-new-variable}. However, Prolog does not work with clauses, but with multi-clauses as a lists. Therefore, we need to do more. 

As in the propositional case, QBF-Prolog treats clauses as multisets, let us call this proof system as $\text{SLD-Q}_{\text{multi}}$-resolution. That is, $\text{SLD-Q}_{\text{multi}}$-resolution proof system is just an \sldQ proof system which retains every multiple copies of any literals within a resolvent clause. For example: ($\leftarrow x, y), (y \leftarrow x, z)~ \derivessldQmulti~\leftarrow x,x,z$, and not $\leftarrow x,z$. Also, QBF-Prolog treats clauses as list of literals. Therefore, the order of literals in any clause matters for QBF-Prolog. To make it clear, let us define the proof system precisely for QBF-Prolog.

\begin{definition}[$\text{SLD-Q}_\text{Prolog}$-resolution]\cite[Chapter 5]{Bunning1999}
An $\text{SLD-Q}_\text{Prolog}$-resolution derivation of a quantified Horn formula is an $\text{SLD}_{\text{multi}}$-resolution derivation in which clauses are regarded and processed as lists. 
In every resolution step, the pivot variable $x$ which is resolved upon must occur as a negative literal at the start of the list for the centre clause ($\neg x, \neg y_1, \neg y_2, \dots \neg y_r $) and as a positive literal at the start of the list for the side clause ($x, \neg s_1, \neg s_2, \dots, \neg s_p$). The resolvent list is formed by concatenating the remainder of the centre clause to the remainder of the side clause ($ \neg s_1, s_2, \dots, \neg s_p,\neg y_1, \neg y_2, \dots, \neg y_r $). There are no additional constraints for the $\forall$-Red rules.
\end{definition}

In QBF-Prolog, not only the order of literals within a clause matters, but also the order of clauses within a  quantified Horn formula matters. A {\bf QBF-Prolog program} $P$ is a quantified Horn formula in which both the clauses and the literals within the clauses are given a fixed ordering. Positive literals always appear in the front of the lists for any definite clause.

Now, let us restate the problems mentioned at the begining of Section~\ref{sec:query-no-new-variable}, with the above discussed modifications.

Instead of a quantified definite Horn formula $\mathcal{F} = \Pi.\phi$, we have a QBF-Prolog program $P = \Pi.\phi$ in which both the clauses and the literals within a clause have fixed ordering. In addition, clauses in $\phi$ are given as a list, with the positive literal in the front. We also have given the query definite Horn clause $C = (x \vee \neg x_1 \vee \neg x_2 \vee \dots \vee \neg x_n)$ as a list of literals, with $var(C) \subseteq var(P)$. We need to answer whether $P=\Pi.\phi \models \Pi.(\phi\wedge C)$? 

Let $i = max_{\ell \in C}\{level(\Pi,\ell)\}$. Lemma~\ref{lemma:no-new-variable} says that our Prolog Algorithm should return a {\bf yes}, for the above problem, if we can search an \sldQ refutation of the following QBF-Prolog program: 
\begin{align}\label{PEquation}
P' = \Pi_i.(\neg x \wedge x_1 \wedge x_2 \dots \wedge x_n \wedge \phi)
\end{align}
Note that all the literals of $C$ becomes existential literals in the above QBF-Prolog program. To be precise, we need to search an $\text{SLD-Q}_\text{Prolog}$-resolution refutation. We first present an exponential time recursive algorithm (Algorithm~\ref{ProloAlgo-for-QBFs}) to search a required $\text{SLD-Q}_\text{Prolog}$-resolution refutation. The Algorithm search a refutation with $\neg x$ as the top clause. That is, it start searching the refutation with negation of the head of the query clause $C$. We call the Algorithm~\ref{ProloAlgo-for-QBFs} with the following input parameters: the QBF-Prolog program $P'$, and $\neg x$ ($\neg x$ will copy in the list $\gamma$).

		\begin{algorithm}[H]
		\DontPrintSemicolon
		\SetKwFunction{FMain}{\bf QBF-Prolog-Recursive-Algo}
		\SetKwProg{Fn}{}{}{}
		\Fn{\FMain{quantified cnf-formula *$P$, list of quantified cnf-variables * $\gamma$}}{
			\KwIn{QBF-Prolog program $P$ as a list of definite quantified Horn clauses, and a goal list $\gamma$ as a list of quantified (negated) literals. For our problem, $\gamma$ is the negation of the head of the query definite clause $C$ and $P$ is equalt to $P'$ (Refer Equation~\ref{PEquation}).}
			\KwOut{\textit{{\bf true}}; if the literals of the goal list are derivable from $P$ using the prolog procedure, \textit{{\bf false}} (or possibly a loop); otherwise}
			\uIf{\textnormal{is-empty}($\gamma$)}{
				return \textit{{\bf true}}}
			\Else{
				\ForEach{clause $C \in P$}{
					\uIf{\textnormal{first-pos-lit}($C$)==\textnormal{first}($\gamma$)}{
						\tcc{\textnormal{first-pos-lit}($C$) returns the only positive literal of $C$, similarly, \textnormal{first}($\gamma$) returns the first existential literal of $\gamma$.~If they are same, call recursively the main program with $P$ as the first argument and the goal list $\gamma'$ as the second argument.~$\gamma'$ is obtained as follows:~Let \textnormal{rest}($C$) outputs exactly the same list except the first positive existential literal.~Similarly,~we have \textnormal{rest}($\gamma$) returns the same list except the first existential literal of $\gamma$.~And\newline $\gamma' = $ \textnormal{append}\big(\textnormal{rest}($C$), \textnormal{rest}($\gamma$)\big), where append concatenates the rest of $\gamma$ list to the rest of $C$ and then removes all the universe literals which are not blocked in the concatenated sequence.}
						\uIf{\textnormal{QBF-Prolog-Recursive-Algo}$(P$, $\gamma')$}{
							\tcc{When the recursive call return sucess, we return success}
							return \textit{{\bf true}}
						}
					}	
				}
				\tcc{For $\gamma$, we considered each clause from $P$ in sequence, but unable to proccess $\gamma$ completely, so return a failure}
				return \textit{{\bf false}}
			}

		}
		\textbf{End Function}
		\caption{Recursive algorithm for Prolog search mechanism for QBFs inspired from \cite{Bunning1999}.}
		\label{ProloAlgo-for-QBFs}
	\end{algorithm}

QBF-Prolog-Recursive-Algo, uses $\text{SLD-Q}_\text{Prolog}$-resolution with a combination of depth first search from left to right with backtracking. Observe that Algorithm~\ref{ProloAlgo-for-QBFs}, may enters an infinite loop, even for the propositional Prolog program.

Infact, due to the deterministic approach, the completeness of the $\text{SLD-Q}_\text{Prolog}$-resolution has lost.

The above QBF-Prolog mechanism can be well explained with the concept of a {\bf refutation tree}~\cite{Bunning1999}. Let $P$ be a QBF-Prolog program and $\gamma$ be the qoal query. The refutation tree $T_P(\gamma)$ is a tree such that its root is labelled with the query clause $\gamma$. Every node labelled with a goal clause $\leftarrow x_1,x_2,\dots,x_n$ has a successor corresponding to every definite Horn clause $x_1 \leftarrow y_1,\dots,y_k \in P$. The successors are labelled with the $\text{SLD-Q}_\text{Prolog}$-resolution resolvent $\leftarrow y_1\dots y_k, x_2 \dots x_n$ in this order.  

For example, Figure~\ref{fig1:refutation-tree} shows a refutation tree for the QBF-Prolog program
$$P = \exists a,b,c,d,e \forall f \exists g.(a \leftarrow e,c,g) \wedge (a \leftarrow d,b) \wedge (d \leftarrow b,f) \wedge (e \leftarrow f) \wedge (b) \wedge (g) $$ and a query clause $\leftarrow a$.

\begin{figure}[h!] 
\centering{
\begin{tikzpicture}[scale=.8, transform shape]

\node[ellipse,draw] (a) at (0, 0) {{\large$\hspace{1mm} \Box \hspace{1mm}$}};
\node[ellipse,draw] (b) at (0, 1.5) {{\large$\leftarrow f$}};
\node[ellipse,draw] (c) at (0, 3.2) {{\large$\leftarrow f, b$}};
\node[ellipse,draw] (d) at (0, 5) {{\large$\leftarrow b,f,b$}};
\node[ellipse,draw] (e) at (0, 6.7) {{\large$\leftarrow d,b$}};
\node[ellipse,draw] (f) at (-3, 8.5) {{\large$\leftarrow a$}};
\node[ellipse,draw] (g) at (-5, 6.7) {{\large$\leftarrow e,c,g$}};
\node[ellipse,draw] (h) at (-5, 5) {{\large$\leftarrow f,c,g$}};
\node[ellipse,draw] (i) at (-5, 3.2) {{\large$\leftarrow f,c$}};

\draw[black, ->] (b) -- (a) node[pos=.5,above, right] {$\forall$-Red steps};
\draw[black, ->] (c) -- (b);
\draw[black, ->] (d) -- (c);
\draw[black, ->] (e) -- (d);
\draw[black, ->] (f) -- (e);
\draw[black, ->] (f) -- (g);
\draw[black, ->] (g) -- (h);
\draw[black, ->] (h) -- (i);

\end{tikzpicture}
  }
  \caption{Refutation tree for QBF-Prolog program $P$ and query $\leftarrow a$ }\label{fig1:refutation-tree}
\end{figure}
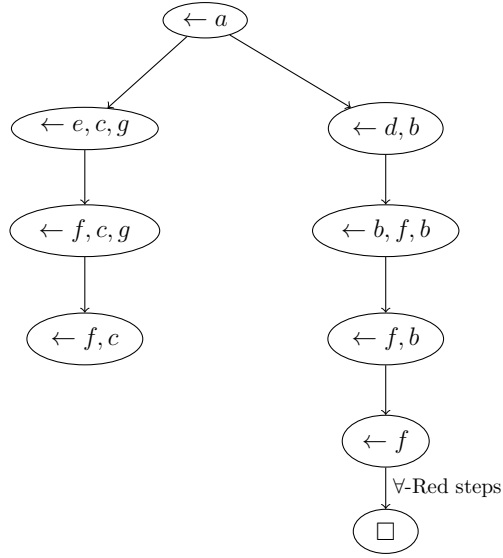

While searching for a refutation, given a goal query $\gamma$ and $P$, three outcomes are possible:
\begin{itemize}
	\item Empty clause $\Box$ is found. We say that the algorithm found the refutation. The result is a {\bf yes}.
	\item The entire tree is searched without finding a node with $\Box$. The result is a {\bf no}, which means that the derivation of the empty clause is not possible.
	\item During the search algorithm follow an infinite branch in the refutation tree. The result is a {\bf loop}.
\end{itemize}

For a QBF-Prolog program $P$ and a Horn query clause $C$, let us define the following function:
	\[
	\text{output}(P,C) = \left\{ 
	\begin{array}{l l}
	\text{yes} & \quad \text{if the algorithm stops with a success}\\
	\text{no} & \quad \text{if the algorithm stops with a failure}\\
	\text{loop} & \quad \text{if the algorithm enters a loop}
	\end{array} \right.
	\]

We next give a linear time algorithm for computing the function $\text{output}(P,C)$.

\subsection{Linear time algorithm for computing the output function}\label{sec:linear-algo-for-qbf-prolog}
Given a QBF-Prolog program $P$ and a query Horn clause $C$, we describe an algorithm which computes the $output(P,C)$ function in linear time, in the length of the inputs. The challenge is to detect whether the algorithm enters into loop in linear time. 

As in the propositional case, we develop QBF-Prolog which is not only interested in finding a refutation, but all possible refutations. Following the ideas from~\cite{Bunning1999}, our proposed algorithm maintains five distinct states for each existential variable of the QBF Prolog program. During the course of processing, each existential variable $x$ is in one of the following states:
\begin{itemize}
	\item $state(x) = new$: the goal $\leftarrow x$ (ie., $\neg x$) has not yet processed.
	\item $state(x) = yes$: the goal $\leftarrow x$ can be refuted in finite number of ways and the refutation tree $T_P(\leftarrow x)$ does not contain an infinite branch.
	\item $state(x) = no$: refutation of the goal $\leftarrow x$ is not possible.
	\item $state(x) = loop$: algorithm follows an infinite branch in the refutation tree $T_P(\leftarrow x)$ while searching for the refuation. That is, algorithm goes into a loop.
	\item $state(x) = inf$: there exists infinitely many ways to refute the goal $\leftarrow x$. Or, after finitely many refutations, the algorithm goes into a loop (follows an infinite branch in the refutation tree).
\end{itemize}   

Any universal variable $u$ of the QBF-Prolog program can be in any one of the following two states. These state depends on the clause in which variable $u$ belongs to.
\begin{itemize}
	\item $blocked(u)=yes$: let $C$ be a clause and $u \in C$, and one cannot apply the $\forall$-red step to $u$ in the clause $C$. To be precise,  there exists an existential literal $\ell \in C$ with $u \leq_{\Pi} \ell$, and $(state(\ell) = no~or~state(\ell) = loop)$. 	

	\item $blocked(u) =no$: let $C$ be a clause and $u \in C$. There does not exists an existential literal $\ell$ with $u \leq_{\Pi} \ell$. That is, one can apply a $\forall$-red step to $u$ in the clause $C$.\\ 
{\bf Or}, for any existential literal $\ell \in C$ with $u \leq_{\Pi} \ell$ we have $state(\ell) =yes$. That is, one can perform the $\forall$-red step on u in future after refuting all the existential literals blocking $u$.
\end{itemize}
Before proceeding further, we present a simple example to clarify the defintion of loop and inf. 

Consider the following slightly modified QBF-Prolog program from~\cite{Bunning1999}: 
$$P = \exists a \forall d \exists b,c.(a \leftarrow d,b,c)\wedge (b) \wedge (b \leftarrow b) $$
We have $state(b) = inf$: since we have a successful refutation of $\leftarrow b$ followed by an infinite branch (loop) in $T_P(\leftarrow b)$. Refer, Figure~\ref{fig2}. 

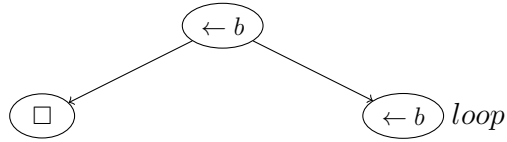
\begin{figure}[h!] 
\centering{
  \begin{tikzpicture}[scale=.8, transform shape]

\node[ellipse,draw,label=right:{\Large$loop$}] (a) at (0, 0) {{\large$\leftarrow b$}};
\node[ellipse,draw] (b) at (-3, 1.5) {{\large$\leftarrow b$}};
\node[ellipse,draw] (c) at (-6, 0) {{\large$\hspace{1mm} \Box \hspace{1mm}$}};

\draw[black, ->] (b) -- (a);
\draw[black, ->] (b) -- (c);

\end{tikzpicture}

  }
  \caption{Refutation tree $T_P(\leftarrow a)$ showing that $state(b) = inf$ }\label{fig2}
\end{figure}

However, we have $state(a) = loop$: since in the refutation tree $T_P(\leftarrow a)$, there exists no refutation for $\leftarrow d,c$ and also has an infinite branch. Refer, Figure~\ref{fig3}.\\

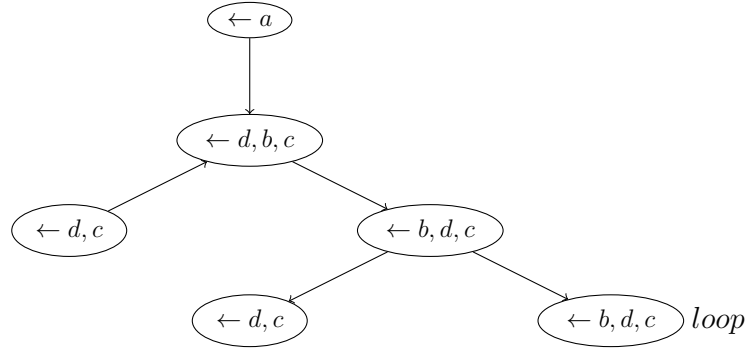
\begin{figure}[h!] 
\centering{
  \begin{tikzpicture}[scale=.8, transform shape]

\node[ellipse,draw,label=right:{\Large$loop$}] (a) at (0, 0) {{\large$\leftarrow b,d,c$}};
\node[ellipse,draw] (b) at (-3, 1.5) {{\large$\leftarrow b,d,c$}};
\node[ellipse,draw] (c) at (-6, 0) {{\large$\leftarrow d,c$}};
\node[ellipse,draw] (d) at (-6, 3) {{\large$\leftarrow d,b,c$}};
\node[ellipse,draw] (e) at (-9, 1.5) {{\large$\leftarrow d,c$}};
\node[ellipse,draw] (f) at (-6, 5) {{\large$\leftarrow a$}};

\draw[black, ->] (b) -- (a);
\draw[black, ->] (b) -- (c);
\draw[black, ->] (d) -- (b);
\draw[black, ->] (e) -- (d);
\draw[black, ->] (f) -- (d);

\end{tikzpicture}

  }
  \caption{Refutation tree $T_P(\leftarrow a)$ showing that $state(a) = loop$ }\label{fig3}
\end{figure}

Observe that given the state of all existentil literals of a clause, computing the states of a universal variable is simple. 
Denote it by {\bf comp-block($u, C$)}.

So, consider the problem of computing $state(x)$ for an existential literal $x$, given a QBF-Prolog program $P$. That is, we need to refute $\leftarrow x$. Clearly, for computing the $\text{SLD-Q}_\text{Prolog}$-resolution refututation of $\leftarrow x$, one has to perform resolution of the clause $\leftarrow x$ with a clause having head $x$. There can be several clauses in $P$ with head $x$. For computing $state(x)$, we need to consider each of them in the order in which they appear in $P$. However, in order to make discussion simple, let us focus on just one clause $C$ with head $x$ and compute the intermediate result $state(x,C)$. We use this result for computing the final one. Let $C = x \leftarrow x_1,x_2,\dots,x_n$ be a clause with head $x$. Note that some of the $x_i$'s may be universal, but $x$ is existential.
\begin{itemize}
	\item $state(x,C)=yes$: if $n=0$, \\ 
{\bf Or}, for all existential $x_i$, $1\leq i \leq n, state(x_i) = yes$ and for all universal variables $x_i$ $1\leq i \leq n, blocked(x_i) = no$.
	\item $state(x,C) = inf$: if for all existential $x_i, 1\leq i \leq n, state(x_i) \in \{yes, inf\}$, and for all universal $x_i, 1\leq i \leq n, blocked(x_i) = no$, and there exist an $x_j, 1 \leq j \leq n, state(x_j) = inf$.

	\item $state(x,C) = loop$: there exists an existential $x_i, 1\leq i \leq n, state(x_i) = loop$, and for all existential $x_k, 1 \leq k < i, state(x_k) = \{yes, inf\}$, and for all universal $x_k, 1 \leq k < i, blocked(x_k) = no$.\\
{\bf Or}, there exists an existential $x_i, 1\leq i \leq n$, such that while refuting $\leftarrow x_i$, the refutation of $\leftarrow x$ is called recursively, and for all existential $x_k, 1 \leq k < i, state(x_k) = \{yes,inf\}$, and for all universal $x_k, 1 \leq k < i, blocked(x_k) = no$.\\
{\bf Or}, there exists $x_i, 2 \leq i \leq n, state(x_i) = no$, and there exists an existential $x_k, 1\leq k < i, state(x_k) = inf$, and for all existential $x_k, 1 \leq k < i, state(x_k) = \{yes, inf\}$, and for all universal $x_k, 1 \leq k < i, blocked(x_k) = no$.

	\item $state(x,C) = no$: there exists an existential $x_i, 1 \leq i < n, state(x_i) = no$, and for all existential $x_k, 1 \leq k < i, state(x_k) = yes$, and for all universal $x_k, 1 \leq k < i, blocked(x_k) = no$.\\
{\bf Or}, there exists a universal variable $x_k, 1 \leq k \leq n, blocked(x_k) = yes$.
\end{itemize}
Using these intermediate results, we are ready to describe the procedure of refuting $\leftarrow x$ with a QBF-Prolog program $P$. That is, procedure to compute $state(x)$, for any existential variable $x$. 

Let the definite Horn clauses $C_x^1, C_x^2,\dots, C_x^m$ be the clauses of $P$ with head $x$.
\begin{itemize}
	\item $state(x) = yes$: there exists an $i, 1 \leq i \leq m, state(x, C_x^i) = yes$, and for all $k, 1 \leq k < i, state(x,C_x^k) = no$, and for all $j, i < j \leq n, state(x,C_x^j) \not \in \{inf, loop\}$.
	
	\item $state(x) = inf$: there exists an $i, 1 \leq i \leq m, state(x,C_x^i) = inf$, and for all $k, 1 \leq k < i, state(x,C_x^k) = no$\\
{\bf Or} there exists an $i, 1 \leq i \leq m, state(x, C_x^i) = yes$, and for all $k, 1 \leq k < i, state(x,C_x^k) = no$, and for all $j, i < j \leq m, state(x,C_x^j) \in \{loop, inf\}$.

	\item $state(x) = loop$: there exists an $i, 1 \leq i \leq m, state(x, C_x^i) = loop$, and for all $k, 1\leq k < i, state(x,C_x^k) = no$.

	\item $state(x) = no$: for all $i, 1 \leq i \leq m, state(x,C_x^i) = no$.
\end{itemize}

Now, we are ready to present our linear time algorithm for computing the function $output(P,C)$. Our algorithm is inspired and a slight modification of the Algorithm from~\cite[Algorithm 5.16]{Bunning1999}.\\

\noindent
{\bf QBF-Prolog-linear: linear time algorithm}\\
Recall the problem: given a QBF-Prolog program $P=\Pi.\phi$ and a query definite Horn clause $C = (x\leftarrow x_1,\dots,x_n) $, with $var(C) \subseteq var(P)$, answer the following:
\begin{align*}
\text{Is }\Pi.\phi \models \Pi.(\phi \wedge C)?
\end{align*}
Following Lemma~\ref{lemma:no-new-variable}, we modify $P$ as $P' = \Pi_i.(\neg x \wedge x_1 \wedge\dots\wedge x_n \wedge \phi)$ and want to find an $\text{SLD-Q}_\text{Prolog}$-resolution of $\leftarrow x$ from $P'$ efficiently. 

In other words, we need to compute $state(x)$. We design the following function for the same:\\ 
{\bf Refutation(x, loc-h(x))}: the function computes $state(x)$. It takes two parameters: variable $x$ and `loc-h(x)': list of clauses from $P'$ with head $x$. It output one of the following: $no, yes, loop,inf$. We maintain these outputs as enum data structure: enum status = $\{$ no, yes, loop, inf $\}$. (Thus, here enumeration variable `no' represents $0$). Finally we set the function $output(P,C)$ as follows:
	\[
	\text{output}(P,C) = \left\{ 
	\begin{array}{l l}
	\text{yes} & \quad \text{if Refutation(x,loc-h(x))}\in \{yes, inf\}\\
	\text{Refutation(x,loc-h-x)} & \quad \text{otherwise}
	\end{array} \right.
	\]

Here is our efficient pseudocode:

\begin{algorithm}[H]
	\DontPrintSemicolon
	\SetKwFunction{FMain}  {status \bf QBF-Prolog-linear}
	\SetKwProg{Fn}{}{}{}
	\Fn{\FMain{quantified cnf-formula *$P$, cnf-clause * $C$}}{
		\KwIn{QBF-Prolog program $P = \Pi.\phi$ as a list of definite quantified Horn clauses, and a definite query clause $C$ as a list }
		\KwOut{output(P,C): yes, if Refutation return $\{$yes, inf$\}$, otherwise return the value of the Refutation function; that is, either loop or no. The return type is the enumeration {\bf status}  }
	\tcc{$C = x \leftarrow x_1,x_2,\dots,x_n$, let $i$ represent the maximum level of the literals of $C$ }
		$P' = \Pi_i.(\neg x \wedge x_1 \wedge\dots\wedge x_n \wedge \phi)$ \\
		\ForEach{existential variable $x \in var(P')$}{ $state(x) = new$}
			result = Refutation(x,loc-h-x)\\
\tcc{\textnormal{Call the Refutation function~[Algorithm~\ref{Algo-Refutation}], which computes $state(x)$ and store the output in the variable  result. Refutation takes the variable x, and a list of clauses from $P'$ with head $x$ as its parameters.}} 
			\uIf{\textnormal{result }$\in \{yes, inf\}$}{
				return \textit{{\bf yes}}
			}
			\Else{
			return ({\bf result})\\
	\tcc{result can be no or loop}	
			}

		}
	\textbf{End Function}
	\caption{Linear time algorithm for QBF-Prolog inspired from \cite{Bunning1999}.}
	\label{linear-ProloAlgo-for-QBFs}
\end{algorithm}

We only need to present the linear time function Refutation. Before presenting the speudocode for the same, we mention some preliminary remarks:

Observe that, for computing $state(x)$ for an existential variable $x$, we need to consider all definite clauses $C_x^i$ of the QBF-Prolog program $P$ with head $x$ in sequence. If $state(x,C_x^i) = no$ for all of them then we need to return a $no$. Otherwise, if we found a $loop$ or an $inf$ at any moment, then we must return immediately a $loop$ or an $inf$ respectively. Only when we encounter a yes for some definite Horn clause $C_x^i$, rest of the clause of $P'$ with head $x$ need to be considered.

During the course of execution, state of variables are going to change. At the begining, we initialize $state(x) = new$ for all existential variables $x$ (refer, Algorithm~\ref{linear-ProloAlgo-for-QBFs}). When we try to find a state of a variable $x$, we replace it's state from new to loop, as the first initialization. This will make sure that the case of a recursive call for determining the state of $x$ is handled correctly. When a first refutation for $\leftarrow x$ is found, we changed the state from loop to inf. This is the second initialization in our Refutation function. Finally, we check that the correct answer is a yes (finitely many refutations) or an inf (infinitely many refutations). 

To make task easy, Refutation function uses another function {\bf Testclause}. \\
{\bf TestClause}: the function takes a definite Horn clause from $P'$, and finds the state of the  positive literal of this clause based on the status of it's negative literals. In other words, it helps the Refutation function for computing $state(x)$ by computing $state(x, C_x^i)$, where $C_x^i$ is a clause from $P'$ with head $x$. TestClause uses Refutation function for the same.

We now present the Refutation function in Algorithm~\ref{Algo-Refutation}. 

\begin{algorithm}[H]
	\DontPrintSemicolon
	\SetKwFunction{FMain}  {status \bf Refutation}
	\SetKwProg{Fn}{}{}{}
	\Fn{\FMain{variable $x$, list of quantified CNF clauses $S$ with head $x$ from $P'$}}{
		\KwIn{variable $x$, and a list of quantified CNF clauses with head $x$  }
		\KwOut{$state(x)$: return type is enumeration status. }
	
			\uIf{$state(x) == new$}{
				\tcc{\textnormal{variable $x$ has not considered yet}}
				\uIf{is-empty$(S)$}{
					\tcc{\textnormal{No clauses with $x$ as it's head}}
					$state(x) =no$ \tcc{\textnormal{we got the answer}}
					
				}
				\Else{
					$state(x) = loop$; \\
					\tcc{\textnormal{First initialization. This handle the case of a recursive call}}
					stop = false;\\
					\While{(not stop) and (not is-empty$(S)$)}{
						result = TestClause(first($S$))\\
						\tcc{\textnormal{first($S$) gives first clause from $S$ with head }$x$}
						\uIf{result == no}{
							$S = S \setminus first(S)$
						}
						\Else{ stop = true}
						
					}
					\uIf{result $\neq$ yes}{
						state(x)=result 
						\tcc{\textnormal{we got the answer}}
					}
					\Else{
						\tcc{\textnormal{test remaining clauses for yes or inf}}
						infflag = false\\
						$S = S \setminus first(S)$; $state(x) = inf$\\
						\While{(not infflag) and (not is-empty$(S)$)}{	
							result = TestClause(first($S$))\\
							infflag = (result == loop) or (result == inf)\\
							$S = S \setminus first(S)$
						}
						\uIf{infflag}{$state(x) = inf$ \tcc{\textnormal{encountered only loop or inf}}}
						\Else{ $state(x) = yes$}					
					}
					
				}
				return $state(x)$
				
			}
			
		}
		\textbf{End Function}
		\caption{Linear time algorithm for QBF-Prolog inspired from \cite{Bunning1999}.}
		\label{Algo-Refutation}
	\end{algorithm}

Now we finish this Section by presenting the function TestClause:

\begin{algorithm}[H]
	\DontPrintSemicolon
	\SetKwFunction{FMain}  {status \bf TestClause}
	\SetKwProg{Fn}{}{}{}
	\Fn{\FMain{CNF clause *C  }}{
		\KwIn{Clause $C = x \leftarrow x_1,x_2,\dots,x_n$ with existential head }
		\KwOut{$state(head(C))$ with respect to the tail of $C$ }
		\tcc{\textnormal{literals of $C$ are arranged as per $\Pi_i$. Initialize as $n=0$ }}
		success = true; infflag = false; result = yes; i =1;\\
	\tcc{\textnormal{compute state of all universal literals of $C$}}
 $\forall i$, with $x_i$ universal, result-array[i]=com-block($x_i, C'$), where $C' = C \setminus x$\\
		\While{(success) and (not is-empty$(C)$)} {
			\uIf{sign(first($C$)) == positive} {
				\tcc{\textnormal{first($C$) returns first literal of $C$, i.e., ignore the head of $C$}}
				$C = C \setminus first(C)$\\

			}
			\Else{
				\tcc{\textnormal{consider next literal from the body of $C$}}
				\uIf{first($C$) == $\forall$}{
					\uIf{result-array[$i$] = yes}{
						result = no \tcc{\textnormal{ answer found, $x_i$ is blocked}}
						break \tcc{\textnormal{ break the while loop}}
					}
				
					 i++ \tcc{\textnormal{$x_i$ not blocked}}
					
				}
				\Else{
				result-array$[$i$]$ = Refutation(first($C$), loc-h(first($C$)))	\\
				result = result-array$[$i$]$; i++;\\
				success = (result == yes) or (result == inf)\\
				\tcc{\textnormal{when result ==loop, success = false and while loop ends}}
				\uIf{result == inf} { infflag = true
				}
				$C = C \setminus first(C)$
				}	
			}
			
		}
		\Switch{result}{
			\Case{result $\in \{yes, inf\}$}{
				\uIf{infflag}{
					return inf
				}
				\Else{ return yes}
			}
			\Case{result == loop}{
				return loop
			}
			\Case{return == no}{
				\uIf{infflag}{
					return loop
				}
				\Else{ return no}
			}

		}

	}
	\caption{TestClause Function computing $state(x, C_x)$ inspired from \cite{Bunning1999}.}
	\label{Algo-TestClause}
\end{algorithm}

\noindent
{\bf Time Complexity Analysis}:\\
The time complexity of Refutation is linear in the length of the QBF-Prolog program $P'$, as every clause of $P'$ is processed at most once. Only we need to be careful in the function TestClause, as we need to compute the state of universal variables as well. For this, we just maintained and array `result-array' and computed the required information at the very begining. Thus, this increases the complexity additively. This proves the following.
\begin{theorem}\label{thm:linear-time}
	For a QBF-Prolog program $P$ and a query clause $C$, with $var(C) \subseteq var(P)$, it is possible to compute the function $output(P,C)$ in linear time in the length of $P$ and $C$. 
\end{theorem}

\section{Theory for QBF-Prolog with no restrictions on the query clause}\label{sec:qbf-prolog-general}
In this Section, we develop theory for the QBF-Prolog without any restrictions on the query Horn clause. That is, the query clause may contains new variables as well. To be precise, given a QBF-Prolog program $P = \Pi.\phi$ and a query Horn clause $C$, is the following true: $$P = \Pi.\phi \models \Pi'.(\phi \wedge C)?$$
where, $C$ may contains new variables which do not occur in $P$. Here, $\Pi'$ is obtained by extending $\Pi$ by the new variables of $C$. There is no restriction on how these new variables are placed in the prefix $\Pi$. 

We show in this section, that this case is not interesting for the QBFs. That is, if the query clause $C$ contains a new existential variable, then $C$ will be implied by $P$ for sure, and if $C$ has a new universal variable $u$, one can drop $u$ and still preserves the answer. 

We show this using the property \QIORplus from \cite{QRATPlus18}. We show that if the query clause contains a new existential variable, then $C$ will satisfy the \QIORplus property for sure and hence $C$ will be redundant for the QBF-Prolog program $P$.  We need the following definitions:

\begin{definition}[Outer clause\cite{QRATOriginal17}]
The outer clause of a clause $C$ on literal $\ell \in C$ with respect to the prefix $\Pi$ is the clause $OC(\Pi, C, \ell) = \{k~|~k \in C, k \leq_\Pi \ell, k \neq \ell\}.$
\end{definition}
Clearly, $OC(\Pi, C, \ell)$ of clause $C$ on literal $\ell \in C$ contains all literals of $C$, exluding $\ell$, which are smaller than or equal to $\ell$ in the variable ordering of prefix $\Pi$. 

\begin{definition}[Outer resolvent\cite{QRATOriginal17}]
Let $C$ be a clause with $\ell \in C$ and $D$ a clause occurring in QBF $\Pi.\phi$ with $ \neg \ell \in D$. The outer resolvent of $C$ with $D$ on literal $\ell$ with respect to the quantifier prefix $\Pi$, denoted $OR(\Pi,C,D,\ell)$ is the following:
$$OR(\Pi,C,D,\ell) = C \setminus \{\ell\} \cup OC(\Pi,D, \neg \ell)$$
\end{definition}

The following Definition is the extension of the property Quantified Implied Outer Resolvent (\QIOR) from~\cite{QRATOriginal17}.
\begin{definition}[\QIORplus\cite{QRATPlus18}]
A clause $C$ has property \QIORplus with respect to QBF $\Pi.\phi$ on literal $\ell \in C$ iff
$$\Pi.\phi \equiv \Pi.(\phi \wedge OR(\Pi,C,D,\ell))$$
for each clause $D \in \phi$ with $\neg \ell \in D$.
\end{definition}
Recall, that by, $\Pi.\phi \equiv \Pi.(\phi \wedge OR(\Pi,C,D,\ell)$, we mean that both QBFs are logical equivalent, that is, every model of one QBF is also the model of another. It has been proved in~\cite{QRATPlus18}, that a clause with \QIORplus property on some existential literal $\ell$ with respect to a QBF $\mathcal{F} = \Pi.\phi$ is redundant for the QBF $\mathcal{F}$. That is, we may add or remove the clause $C$ from $\mathcal{F}$ without effecting the satisfiability of $\mathcal{F}$. We below state the corresponding Theorem from~\cite{QRATPlus18}.

\begin{theorem}[\cite{QRATPlus18}\label{thm:qiorplus}]
Given a QBF $\mathcal{F} = \Pi.\phi$ and a clause $C$ with \QIORplus on an existential literal $\ell \in C$ with respect to QBF $\mathcal{F}' = \Pi'.\phi'$, where $\phi' = \phi \setminus \{C\}$ and $\Pi'$ is same as $\Pi$ with variables and respective quantifiers removed that no longer appear in $\phi'$. Then $\mathcal{F} \equiv_{\text{sat}} \mathcal{F}'$.
\end{theorem}
Thus, Theorem~\ref{thm:qiorplus} states that the clause $C$ with \QIORplus property on some existential literal is redundant for the QBF, and can be added or removed. We also have universal elimination theorem from~\cite{QRATPlus18}.

\begin{theorem}[\cite{QRATPlus18}\label{thm:qiorplusuniversal}]
Given a QBF $\mathcal{F}_0 = \Pi.\phi$ and $\mathcal{F} = \Pi.(\phi \cup \{C\})$ where $C$ has \QIORplus on a universal literal $\ell \in C$ with respect to the QBF $\mathcal{F}_0$. Let $\mathcal{F}' = \Pi.(\phi \cup \{C'\})$, with $C' = C \setminus \{\ell\}$. Then $\mathcal{F} \equiv_{\text{sat}} \mathcal{F}'$.  
\end{theorem}


{\bf Remarks}: Although, \QIORplus is extremely powerful in terms of redundancy detection, checking for the \QIORplus property in practice is too costly. Since, even for the propositional case checking whether an outer resolvent is implied by a propositional formula is co-NP hard. So, one needs a redundancy property that can be checked in polynomial time. One such redundancy property is the \QRAT~\cite{QRATOriginal17}. Following \QRAT, \QRATplus~\cite{QRATPlus18} have been introduced, which is a polynomial time redundancy property for a clause, based on QUP (quantified unit propogation). It has been proved in~\cite{QRATPlus18}, that \QRATplus is more powerful than \QRAT in terms of redundancy detection. 

However, we do not need \QRATplus to prove our results, \QIORplus will suffice. 

\begin{observation}
Let $P = \Pi.\phi$ be a QBF Prolog program. And, let $C$ be a definite Horn query clause. If there exists an existential literal say, $\ell \in C$ which is new. That is, $\ell \in var(C)$ but $\ell \not \in var(\phi)$. Then, we know that $C$ for sure has \QIORplus property with respect to the QBF Prolog program $P$: since there exists no clause $D$ in $\phi$ with $\neg \ell \in D$. Therefore, $C$ is  redundant by Thereom~\ref{thm:qiorplus}. That is, the answer is yes for our problem.
\end{observation}

\begin{observation}
Let $P = \Pi.\phi$ be a QBF Prolog program. And, let $C$ be a definite Horn query clause. If there exists a universal literal $\ell \in C$ which is a new literal, that is universal $\ell \in var(C)$ but $\ell \not \in var(\phi)$. Then by Thereom~\ref{thm:qiorplusuniversal}, we can drop universal literal $\ell$ from $C$ safely and then ask the problem for the remaining clause.
\end{observation}

\noindent
{\bf QBF-Prolog for First-order Predicate Horn Logic}:
Recall the Definition of first-order predicate Horn logic from Section~\ref{sec:notations}. In this logic, we may replace each literal in a Horn clause by arbitrary predicate of any bounded arity. The QBF-Prolog defined in the previous section is capable of solving the first-order Horn logic. That is, given facts, rules and query as first-order predicate Horn formulas, the QBF-Prolog computes the function output(P,C) in linear time. Note that a predicate may contains arbirary quantified variables, which is not allowed in the existing Prolog. 
The algorithm for first-order predicate Horn logic is exactly as the Algorithm~\ref{linear-ProloAlgo-for-QBFs}. The only difference is that now the algorithm treats each predicate of the query Horn clause as a literal. For the resolution step of literals corresponding to the predicates, the algorithm may use unification to make them syntatically equal.

To be precise, let $C$ be a first-order query Horn clause. Every predicates $P(x_1,\cdots,x_k)$ appearing in clause $C$ is considered as a literal $\ell$. 
Also, the literal $\ell = P(x_1,\cdots,x_k)$ is considered as an existential literal (recall the definition of abstraction). 
We can resolve the literal $P(x_1,\cdots,x_k)$ with the same predicate $P(y_1,\cdots,y_k)$ in rule clause $\tilde{C}$ only if for each index $i$,
\begin{enumerate}
\item[(1)] the quantifier of $y_i$ is $\forall$, or,
\item[(2)] if quantifier of $y_i$ is $\exists$ then the quantifier of $x_i$ is $\exists$.
\end{enumerate} 
Recall that the variables occur in the same order as in the quantifier prefix. 
As the quantification of variables stays same in the resolvent, it can restricts the set of clause with which one can resolve a predicate $P$.

Note that for any literals $P(x)$ and $P(y)$ with $x\neq y$, separate states will be maintained by the algorithms.

\section{Applications of QBF-Prolog}\label{sec:application-of-qbf-prolog}
In this Section, we present some applications of the QBF-Prolog. Observe that the existing Prolog is a subset of QBF-Prolog. Hence, QBF-Prolog can handle problems that can be handled by the existing Prolog. Now we present some examples which can be handled by QBF-Prolog but not by the exisiting Prolog.
\subsection{Application 1: QBFs} \label{subsec:app1}
Recall the limitations of Prolog: it is not possible to define rules having an existentially quantified head, and having body with arbitrary quantifications. 

Below example illustrates, that the QBF-Prolog has no such restrictions. Let us consider the following simple QBF-Prolog program~\cite{BeyersdorffCMS18-TOCL}:
\begin{align}
\notag \mathcal{P}_n = \Pi.\phi =&\exists e_0 \forall u_1 \exists e_1 \dots u_n \exists e_n.  \\
\notag \textrm{For~} i \in [n], D_i&: ~~~~ ({e}_{i-1} \vee \neg u_i \vee \neg e_i) \ \wedge \\
\notag D_{n+1}&: ~~~~ ({e}_n)
\end{align}

Clearly, in this QBF-Prolog program, the head ($e_{i-1}$) of the rules are existential, which is not possible in Prolog. Also, the body of the Horn clauses have both existentially (i.e., $e_i$'s) and universally (i.e., $u_i$'s) quantified variables. 

Consider a simple query clause $C = (e_0)$. 
Observe that $var(C) \subseteq var(\mathcal{P}_n)$. As, the variable $e_0$ belongs to the first level in the quantifier prefix $\Pi$, we have the abstraction $\Pi_1.\phi$ is same as $\Pi.\phi$. Since, we have simple \sldQ refutation of $\Pi_1.(\neg e_0 \wedge \phi)$ (Figure~\ref{fig-refutation}), the QBP-Prolog is capable to show that the following: 
$\Pi.\phi \models \Pi.(\phi \wedge C)$.

To be precise, in order to solve the above problem, we invoke the Algorithm~\ref{linear-ProloAlgo-for-QBFs}, with parameters $\mathcal{P}_n$ and $(e_0)$. The Algorithm inturn calls the Algorithm~\ref{Algo-Refutation}, with parameters $e_0$ (precisely, $\neg e_0$) and $(e_0 \vee \neg u_1 \vee \neg e_1)$. Algorithm~\ref{Algo-Refutation} inturn call the Algorithm~\ref{Algo-TestClause}, with the parameter $(e_0 \vee \neg u_1 \vee \neg e_1)$. The recursion will end when the TestCase($(e_n)$) function returns a yes, which evetually reaches to the first level of recursion and Algorithm~\ref{linear-ProloAlgo-for-QBFs} also returns a yes.
\begin{figure}[h!] 
\centering{
\begin{tikzpicture}[scale=.8, transform shape]
\node[ellipse,draw] (a) at (0, 0) {{\large$\hspace{1mm} \Box \hspace{1mm}$}};
\node[ellipse,draw] (b) at (0, 2) {{\large$\neg u_1 \vee \cdots \vee\neg u_n$}};
\node[ellipse,draw] (c) at (-3.5, 3.5) {{\large$\neg u_1 \vee \cdots \vee \neg u_n \vee \neg e_n$}};
\node[ellipse,draw,label=right:{\Large$D_{n+1}$}] (d) at (2, 3.5) {{\large$\hspace{1mm} e_n \hspace{1mm}$}};
\node[ellipse,draw] (e) at (-5.5, 5.5) {{\large$\neg u_1 \vee \neg u_2 \vee \neg u_3 \vee \neg e_3$}};
\node[ellipse,draw] (f) at (-7.8, 7) {{\large$\neg u_1 \vee \neg u_2 \vee \neg e_2$}};
\node[ellipse,draw,label=right:{\Large$D_3$}] (g) at (-1.8, 7) {{\large$ e_2 \vee \neg u_3 \vee \neg e_3$}};
\node[ellipse,draw] (h) at (-9.8, 8.5) {{\large$\neg u_1 \vee \neg e_1$}};
\node[ellipse,draw,label=right:{\Large$D_2$}] (i) at (-5, 8.5) {{\large$ e_1 \vee\neg u_2 \vee\neg e_2 $}};
\node[ellipse,draw,label=left:{\Large$\neg C$}] (j) at (-11.5, 10) {{\large$\hspace{1mm} \neg e_0 \hspace{1mm}$}};
\node[ellipse,draw,label=right:{\Large$D_1$}] (k) at (-7.3, 10) {{\large$ e_0 \vee \neg u_1 \vee \neg e_1 $}};

\draw[black, ->] (b) -- (a) node[pos=.5,above, right] {$n~\forall$-Red steps};
\draw[black, ->] (c) -- (b);
\draw[black, ->] (d) -- (b);
\draw[black,dashed, ->] (e) -- (c);
\draw[black, ->] (f) -- (e);
\draw[black, ->] (g) -- (e);
\draw[black, ->] (h) -- (f);
\draw[black, ->] (i) -- (f);
\draw[black, ->] (j) -- (h);
\draw[black, ->] (k) -- (h);

\end{tikzpicture}
  }
  \caption{An {\sldQ} refutation of $\Pi.(\neg e_0 \wedge \phi)$~\cite{BeyersdorffCMS18-TOCL} for Application~\ref{subsec:app1} }\label{fig-refutation}
\end{figure}
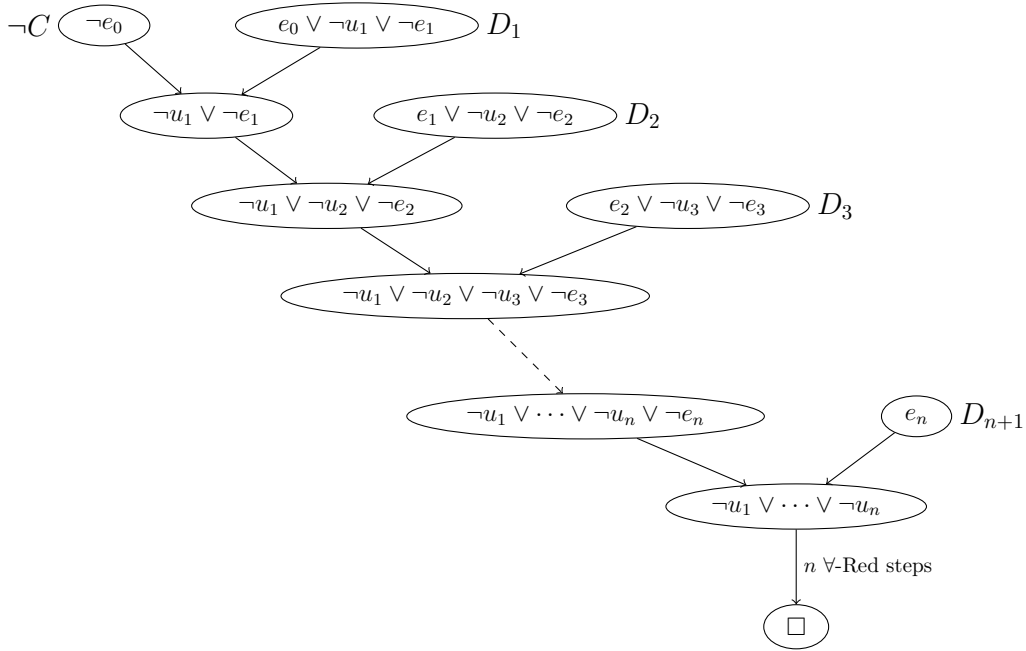

\subsection{Application 2: Tree-bipartite-problem}\label{subsec:tree-bipartite}
Consider the problem tree-bipartite from Example~\ref{example:nocycle-tree} of Section~\ref{subsec:contri}. The problem has the following rules: undirected graphs without cycles are bipartite, and
trees are acyclic. Such rules can be easily encoded in the Existing Prolog. Let us present an encoding in detail. The encoding uses the following predicates with the following interpretations:
\begin{align*}
\notag & Nocycle(g)~~/*\text{the predicate is 1 iff } g \text{ has no cycles}*/\\
\notag & Bipartite(g)~~/*\text{ the predicate is 1 iff } g \text{ is bipartite}*/\\
\notag & Tree(g)~~/*\text{ the predicate is 1 iff } g \text{ is a tree}*/
\end{align*}
Now, using these predicates, we have the following QBF-Prolog program for the tree-bipartite problem.
\begin{align*}
\notag & \forall g. Nocycle(g) \rightarrow Bipartite(g)~~/*\text{ rule 1}*/\\
\notag & \forall g. Tree(g) \rightarrow Nocycle(g)~~/*\text{ rule 2}*/
\end{align*}
Consider the following query which cannot be handled by the existing Prolog as they required universally quantified variables: are all trees bipartite? Below is the query clause for the same.
\begin{align*}
\notag & \forall g. Tree(g) \rightarrow Bipartite(g)~~/*\text{ query}*/\\
\end{align*}
Our algorithm solves the problem as follows: It first adds the negation of the query clause in the QBF-Prolog program. That is we have,
\begin{align*}
\notag & P' = \Pi_i. \neg Bipartite(g) \wedge Tree(g) \wedge \phi
\end{align*}
Certainly, the predicates corresponding to the query clause becomes existential. The algorithm picks the $\neg Bipartite(g)$ as the top clause and derives the empty clause using rule 1, and rule 2, as follows:
 \begin{prooftree}
 \AxiomC{$\neg Bipartite(g)$}
 \AxiomC{$\neg Nocycle(g) \vee Bipartite(g)$}
 \BinaryInfC{$\neg Nocycle(g)$}
 \AxiomC{$\neg Tree(g) \vee Nocycle(g)$}
 \BinaryInfC{$ \neg Tree(g)$}
 \AxiomC{$Tree(g)$}
 \BinaryInfC{$ \Box$}

 \end{prooftree}

\subsection{Application $3$: Bipartite problem}\label{subsec:app2}
Consider the problem related to the simple bipartite graphs $G=(V,E)$ with vertex partitions $U$ and $W$, from Section~\ref{subsec:contri} (refer~Example~\ref{example-bipartite}). Assume that $G$ has no loops and multiple edges. Consider a QBF-Prolog program which specifies the class of bipartite graphs. This can be encoded in first-order predicate Horn logic with only universally quantified variables. Let us present below one of the encodings in detail. The encoding uses the following predicates with the following interpretations:
\begin{align}
\notag & E(x,y)~~/*\text{the predicate is 1 iff there exists and edge between vertices }x \text{ and } y*/\\
\notag & first(x)~~/*\text{ the predicate is 1 iff } x \text{ belongs to the first part } U \text{ of the bipartite graph}*/\\
\notag & second(x)~~/*\text{ the predicate is 1 iff } x \text{ belongs to the second part } W \text{ of the bipartite graph}*/
\end{align}
Now, using these predicates, we have the following QBF-Prolog program for the bipartite graphs. 
\begin{align*}
\notag & \forall x, y. [first(x) \wedge E(x,y) ] \rightarrow second(y)~~/*\text{ rule 1}*/\\
\notag & \forall x, y. [E(x,y) \wedge second(x)] \rightarrow first(y)~~/*\text{ rule 2}*/\\
\notag & \forall x. first(x) \rightarrow \neg second(x)~~/*\text{ rule 3}*/ \\
\notag & \forall x. E(x,y) \rightarrow E(y,x)~~/*\text{ rule 4}*/ 
\end{align*}
We quickly explain the above rules: Rule 1 says that if an edge $\{x,y\}$ is present in the bipartite graph $G$ and $x \in U$ then we have $y\in W$. Similarly, rule 2 can be explained. Rule 3 says that if a vertex belong to $U$ then it does not belong to $W$. Rule 4 says that the predicate $E(x,y)$ is symmetric in nature.

Now, consider the following query regarding a bipartite graph $G=(\{U,W\},E)$ which uses universal variables and hence cannot be supported by the existing Prolog: starting from any vertex $x \in U$, and after jumping two hops via edges in $G$, do we again reach the set $U$? Below is the quantified query clause $C$ for the same:
\begin{align}
\notag & \forall x, y, z. [first(x) \wedge E(x,y) \wedge E(y,z)] \rightarrow first(z)~~/*\text{ query Horn clause for QBF-Prolog}*/
\end{align}
Certainly, such queries are valid for the QBF-Prolog. The algorithm solves the problem as follows: It first adds the negation of the query clause $C$ in the QBF-Prolog program. That is, we have
$$P' = \Pi_i. \neg first(z) \wedge first(x) \wedge E(x,y) \wedge E(y,z)  \wedge \phi$$
Clearly, all the literal in $P'$ corresponding to the predicate of $C$ becomes existential. The algorithm picks $\neg first(z)$ and resolves it with rule 2 after substituting $y/z$ in the same. Substituting $y/z$ in rule 2 we get: $\neg E(x,z) \vee \neg second(x) \vee  first(z)$. That is, we have
 \begin{prooftree}
 	\AxiomC{$\neg first(z)$}
 	\AxiomC{$\neg E(x,z) \vee \neg second(x) \vee first(z)$}
 	\BinaryInfC{$\neg E(x,z/y) \vee \neg second(x)$}
 \AxiomC{$E(x,y)$}
 \BinaryInfC{$\neg second(x)$}

 \end{prooftree}
We derived $\neg second(x)$ at this moment. We proceed as follows: 
 \begin{prooftree}
 \AxiomC{$\neg second(x/y)$}
 \AxiomC{$\neg first(x) \vee \neg E(x,y) \vee second(y)~~\text{ /*rule 1*/}$}
 \BinaryInfC{$\neg first(x) \vee \neg E(x,y)$}
 \AxiomC{$first(x)$}
 \BinaryInfC{$ \neg E(x,y)$}
 \AxiomC{$E(x,y)$}
 \BinaryInfC{$ \Box$}

 \end{prooftree}
Thus, the algorithm returns true for this query.

\subsection{Application 4: Simple relations}\label{subsec:simple-relations}
Consider the problem of Example~\ref{example:simple-relations} of the Section~\ref{subsec:contri}. The problem uses the following predicates: 
\begin{align*}
\notag & P(h,k)~~\text{/* the predicate is 1 iff } k = 2h */\\
\notag & R(h,k)~~\text{/* the predicate is 1 iff } h<k */
\end{align*}
We may give several interpretations to these predicates. For example, if $k$ and $h$ represent graphs, then the predicate $R(k,h)$ is $1$ iff the graph $h$ is a subgraph of the graph $k$. The predicate $P(k,h)$ is $1$ iff the graph $k$ is a superset of the graph $h$. That is, the graph $k$ is constructed from $h$ by say adding a vertex.

Let us now consider the following facts and rules:
\begin{align*}
\notag & \forall h \exists k. P(h,k)~~\text{/* fact 1 */}\\
\notag & \forall h, k. P(h,k) \rightarrow R(h,k)~~\text{/* rule 1 */}\\
\notag & \forall h_1,h_2,h_3. [R(h_1,h_2) \wedge R(h_2,h_3)] \rightarrow R(h_1,h_3)~~\text{ /* rule 2 */}
\end{align*}
Now consider the following query:
\begin{align*}
\notag & \forall h \exists g,k. [ P(h,g) \wedge P(g,k)] \rightarrow R(h,k)~~\text{ /* query */}
\end{align*}
Clearly the query is correct according to our interpretation. Our algorithm solves the problem as follows: It first adds the negation of the query clause in the QBF-Prolog program. That is, we have,
\begin{align*}
\notag & P' = \Pi_i. \neg R(h,k) \wedge P(h,g) \wedge P(g,k) \wedge \phi
\end{align*}
Recall, all predicates of the query clause becomes existential. The algorithm picks $\neg R(h,k)$ as the top clause and derive the empty clause as follows:
 \begin{prooftree}
 \AxiomC{$\neg R(h,k)$}
 \AxiomC{$\neg P(h,k) \vee R(h,k)\text{/* rule 1 */}$}
 \BinaryInfC{$\neg P(h,k)$}
 \AxiomC{$P(h,g/k)\text{/* added clause from }P' */$}
 \BinaryInfC{$ \Box$}

 \end{prooftree}

\subsection{Application 5}\label{sec:tree-encoding}
Consider the Prolog program mentioned in Example~\ref{example-tree}. The problem cannot be handeled by existing Prolog, as it requires variables with both quantifiers. Also, the proof for the problem is based on induction, as a result, instead of solving the problem, the QBF-Prolog return a loop.  

Let us now encode the problem: given trees with unique root and such that every node has only one parent node. The program uses the following predicates:
\begin{align*}
\notag & E(x,y)~~/*\text{the predicate is 1 iff there exists and edge between vertices }x \text{ and } y*/\\
\notag & Notroot(x)~~/*\text{ the predicate is 1 iff } x \text{ is not the root}*/\\
\notag & Root(x)~~/*\text{ the predicate is 1 iff } x \text{ is the root}*/\\
\notag & P(x)~~/*\text{ the predicate is an indicator used for connectivity}*/
\end{align*}
Using the above predicates we now present a QBF-Prolog program specifying the trees.
\begin{align*}
\notag & \forall x, y. E(x,y) \rightarrow \neg E(y,x)~~/*\text{ rule 1: tree is directed}*/\\
\notag & \forall x \exists y. Notroot(x) \rightarrow E(y,x)~~/*\text{ rule 2: if x is not root then it has a parent y}*/\\
\notag & \forall x \exists y \forall z. E(y,x) \wedge Notroot(x) \wedge \neg(z=y) \rightarrow \neg E(z,x)~~/*\text{rule3:  every non root node has unique parent}*/\\
\notag & \forall x. Root(x) \rightarrow \neg Notroot(x)~~/*\text{ rule4: if x is a root then x is not a Nonroot}*/\\
\notag & \exists x \forall y Root(x) \wedge \neg(y=x) \rightarrow \neg Root(y)~~/*\text{ rule5: there is a unique root}*/\\
\notag & \exists x. Root(x) \rightarrow P(x)~~/*\text{rule6:  if x is a root then the indicator P(x) is true}*/\\
\notag & \forall x,y. P(x) \wedge E(y,x) \rightarrow P(y)~~/*\text{rule7: this rule is for connectivity}*/ 
\end{align*}
Now, consider the following query: 
\begin{align*}
\notag & \forall x. P(x)~~/*\text{ is tree connected?}*/
\end{align*}
Observe that one may easily prove this query via induction. However, QBF-Prolog will detect a loop since it unables to mimic the inductive proofs. To be precise, the algorithm adds the negation of the predicate of the query clause in the QBF-Prolog program. That is, we have
$$P' = \Pi_i. \neg P(x) \wedge \phi$$
So we have
 \begin{prooftree}
 \AxiomC{$\neg P(x)$}
 \AxiomC{$\neg P(x) \vee \neg E(x,y/x) \vee P(y/x)~~\text{ rule 7: substitute y by x}$}
 \BinaryInfC{$\neg P(x) \vee \neg E(x,y/x)$}
 \end{prooftree}
The algorithm again needs to refute $\leftarrow P(x)$ hence outputs a loop. 
\section{Conclusion}\label{sec:conclusion}
The paper overcomes one of the major logical limitations of Prolog of not allowing arbitrary quantified variables in the rules, facts and queries. The paper achieves this by extending the \sld proof systems to quantified Boolean Horn formuals, followed by proposing an efficient implementation for the following problem: given a quantified Boolean Horn formula $\mathcal{F} = \Pi.\phi$ and a clause $C$, is $\Pi.\phi \implies \Pi.(\phi \wedge C)?$. 
The paper shows that the implementation can also handles the first-order predicate Horn logic. 

We also saw that the proposed algorithms unables to solve problems with inductive proofs. The next step is to extend the algorithm for simple problems based on induction. For example, extending the algorithm to handle the tree problem of Application~\ref{sec:tree-encoding}.

The satisfiability problem for the first-order predicate formulas is known as the Satisfiability Modulo Theories (SMT) problem.  
Designing efficient SMT solvers for formulas over arbitrary quantified variables is challenging. 
Our implementation behaves as an SMT-solver for the first-order predicate Horn formulas, over arbitrary quantified variables. 


{\subsubsection*{Acknowledgements.}
We thank Olaf Beyersdorff, Meena Mahajan, and Leroy Chew for useful discussions on extending Prolog to QBFs in the Dagstuhl Seminar `SAT and Interactions' (20061). This work is partially supported by FONDECYT-CONICYT postdoctorate program project: 3190527 in the Department of Mathematics, Pontifica Universidad Catolica de Chile ($1$st author).
}

\bibliographystyle{plain}
\bibliography{myrefs}

\end{document}